\documentclass[10pt]{article}

\usepackage[preprint]{neurips_2024}

\usepackage[utf8]{inputenc} 
\usepackage[T1]{fontenc}    
\usepackage{hyperref}       
\usepackage{url}            
\usepackage{booktabs}       
\usepackage{amsfonts}       
\usepackage{nicefrac}       
\usepackage{microtype}      
\usepackage{xcolor}         
\usepackage{subcaption}
\usepackage{caption}
\usepackage{amsmath,amsthm, amssymb}
\usepackage{graphicx}
\usepackage{mathtools}

\newtheorem{assumption}{Assumption}

\newtheorem{theorem}{Theorem}
\newtheorem{corollary}{Corollary}
\newtheorem{proposition}{Proposition}

\theoremstyle{remark}
\newtheorem*{remark}{Remark}
\DeclareMathOperator*{\argmin}{argmin}

\title{Data-adaptive exposure thresholds for the Horvitz-Thompson estimator of the Average Treatment Effect in experiments with network interference}

\author{%
  Vydhourie Thiyageswaran \\
  Department of Statistics\\
  University of Washington\\
\And 
Tyler McCormick \\
  Department of Statistics\\
  University of Washington\\
  \And
  Jennifer Brennan \\
  Google Research \\
}

\begin{document}

\maketitle


\newcommand{\theHalgorithm}{\arabic{algorithm}}





\begin{abstract}
    Randomized controlled trials often suffer from interference, a violation of the Stable Unit Treatment Values Assumption (SUTVA) in which a unit's treatment assignment affects the outcomes of its neighbors. This interference causes bias in naive estimators of the average treatment effect (ATE). A popular method to achieve unbiasedness is to pair the Horvitz-Thompson estimator of the ATE with a known exposure mapping: a function that identifies which units in a given randomization are not subject to interference. For example, an exposure mapping can specify that any unit with at least $h$-fraction of its neighbors having the same treatment status does not experience interference. However, this threshold $h$ is difficult to elicit from domain experts, and a misspecified threshold can induce bias. In this work, we propose a data-adaptive method to select the "$h$"-fraction threshold that minimizes the mean squared error of the Hortvitz-Thompson estimator. Our method estimates the bias and variance of the Horvitz-Thompson estimator under different thresholds using a linear dose-response model of the potential outcomes. We present simulations illustrating that our method improves upon non-adaptive choices of the threshold. We further illustrate the performance of our estimator by running experiments on a publicly-available Amazon product similarity graph. Furthermore, we demonstrate that our method is robust to deviations from the linear potential outcomes model.
\end{abstract}


\section{Introduction}
Under network interference, where the Stable Unit Treatment Values Assumption (SUTVA) is violated, estimating average treatment effects incorporating direct effects (impact of the treatment on the respondent's outcome) and indirect effects (impact of treated peers on the respondent's outcome) is a central problem in learning causal effects. We study the setting of randomized controlled trials. For example, estimating the adoption of a new idea through random assignment of people to advertisements is complicated by the dissemination of information through social interactions. In settings where user behavior informs a prediction algorithm, the behavior of other users creates spillovers through their impact on the prediction algorithm.  As such discrete interactions can be modeled through network interference models, the problem of estimating average treatment effects under network interference becomes a ubiquitous one. 

A major challenge arises from the lack of knowledge of the exact interference structure that contributes to the indirect effect.  For some treatments, having treated peers might have very little impact (e.g. a drug for a noncommunicable disease that's only available as part of a clinical trial), whereas for others the impact could be extremely large (e.g. vaccines).  \emph{Exposure maps} classify individuals into groups with similar contact patterns with treated peers.  

One simple form of an exposure map classifies individuals as indirectly exposed if more than some fixed fraction $h$ of their network is treated, meaning that the researcher expects contact with a smaller fraction than $h$ to have no impact on the outcome and a fraction larger than $h$ to have a substantial, but relatively consistent impact as the fraction increases above the threshold. Our setting aligns with the fractional thresholds model \citep{watts1999small}, in which both treated and untreated peers influence an individual's outcome, but in opposing directions. \citet{centola2007complex} illustrate this with the example of refraining from littering in a neighborhood: an individual’s decision to avoid littering depends on the relative number of neighbors who also refrain from littering versus those who do not. As the neighborhood size increases, a greater number of non-littering neighbors is required to ``activate" the individual's own decision to refrain from littering. 



This notion of exposure is especially relevant in the study of technology adoption, particularly in contexts where a technology spreads only after a sufficient number of individuals in the network have adopted it — reaching what is known as the threshold of adoption~\citep{acemoglu2011diffusion,reich2016diffusion}. For instance, a messenger app may only gain traction once a critical fraction of social contacts start using it, surpassing the adoption threshold necessary for further diffusion.

Given a known $h$-fractional exposure mapping, the Horvitz-Thompson estimator under this criterion is commonly used as an unbiased estimator of the average treatment effect. 
Without domain knowledge of the interference structure, however, an estimator with fixed treatment exposure conditions for this estimand will be high in either bias or variance. One could, for example, resort to using the Horvitz-Thompson estimator at extreme $h$s for unbiasedness, such as when the unit in question, and \textit{all} its neighbors are treated. However, this incurs the cost of a large variance, as the estimator scales inversely with the probability of the unit and all its neighbors being treated. This probability can be very small for high-degree nodes, particularly, in dense graphs and when treatment fractions are small.  Choosing a small threshold, however, means that many individuals who are in reality not impacted indirectly by the treatment will be counted amongst the indirectly exposed, biasing effect estimates. 
This motivates the study of mean-squared-error (MSE)-minimizing exposure threshold selection to combine with existing estimators. 

Our contribution to this setting is to propose a simple data-dependent approach to estimating the average treatment effect optimally in terms of the MSE in the finite population setting. We compute an approximate rate of change of bias and variance with respect to the threshold to inform the MSE-optimal threshold for the Horvitz-Thompson estimator. In particular, we propose using linear regression for an approximate bias-rate of change measure as a simple and intuitive approach. We provide a non-asymptotic concentration bound on the estimated MSE-optimal threshold with respect to the MSE-optimal threshold under best average linear fit. Additionally, this study allows us to understand graph structures under which this approximation is more precise.\subsection{Related Works}

The exposure mapping framework is a common strategy for quantifying spillovers. In \citep{aronow2017estimating, ugander2013graph, sussman2017elements, hardy2019estimating}, and \citep{auerbach2021local}, the authors expounded upon different exposure mappings, and estimation under these settings. \citep{eckles2017design, toulis2013estimation} assume a (weighted) fraction of treated neighbors in characterizing the exposure. We work under a similar setting.
Other common assumptions include exposures described by the raw count numbers of the treated/control neighbors \citep{ugander2013graph}. In \citep{basse2018model, cai2015social}, a  (known) generalized linear form of network effects on neighborhood treatments was assumed. We do not assume this. Related to our work, in \citep{zhu2024integrating}, the authors fit functionals on the treatment assignment and exposure.

In this paper, we focus on the Horvitz-Thompson estimator (and extend to the Difference-in-Means estimator in Appendix \ref{appendix:difference-in-means}). In \citep{chin2019regression}, the author proposes regression adjustment estimators that predict potential outcomes under global treatment and control conditions. Unlike traditional regression adjustments, this approach constructs adjustment variables from functions of the treatment assignment vector, framing the learning of a more flexible ``exposure mapping" as a feature engineering problem.

Our approach at the estimation stage can loosely be viewed as the ``dual" of choosing cluster sizes in cluster-randomization approaches such as \citep{ugander2013graph, eckles2017design} at the design stage. 

For off-policy evaluation in reinforcement learning, \citep{su2020adaptive} investigate adaptive MSE-optimal kernel bandwidth selection (a widely studied problem, see for example, \citep{fan1992variable}, \citep{ruppert1997empirical} \citep{kallus2018policy}) using Lepski's method \citep{goldenshluger2011bandwidth}. To the best of our knowledge, a reframing of our problem in terms of optimal bandwidth selection \ref{appendix:bandwidth} has not been studied in the average treatment effect estimation under network interference setting. The authors of \citep{belloni2022neighborhood} have a different, but related, goal of estimating the direct effect under network interference. Their approach was to first estimate interference structure by estimating the $m_i$-hop influence (beyond the immediate graph neighborhood), or exposure radius under assumptions of additivity of the main effects in the potential outcome model.

\subsection{Notation.}
 We use $h$ to denote the threshold for ``treatment exposure", and $1-h$ for ``control exposure." Therefore, for larger $h$, i.e. $h$ closer to one, we have a more restrictive setting, where only the subset of treated (resp. control) units with most of their neighbors treated (resp. control) are counted as treatment (resp. control) exposed. For smaller $h$, i.e. $h$ close to zero, we have a less restrictive setting, as a larger subset of treated (resp. control) units satisfy the threshold condition. Throughout the paper we use $d_i$ for the degree of node $i$, and $d$ when the graph is regular.

\section{Setup}
We study the finite population setting. Additionally, we consider a fractional exposure mapping. That is, for any treatment assignment vector $\mathbf{Z} \in \{0,1\}^n$, the effective influence of the graph's treatment assignment on the outcome of node $i$ is equivalent to the influence of fractional treatment assignments in the neighborhood of node $i$. This reduces our potential outcome model that so that it takes the following form:
\begin{align}
\label{potential outcome}
    y_i(\mathbf{Z}) = y_i(z_i, e_i) =& \alpha + g(z_i) + f(e_i) + \epsilon_i\\ &\equiv \alpha + \psi(z_i, e_i) + \epsilon_i,
\end{align} where $z_i \in \{0,1\}$ is the treatment assignment of unit $i$, $e_i \in [0,1]$ is the fraction of treated neighbors of unit $i$ (not including $i$ itself), and the $\epsilon_i$ are uniformly-bounded and non-random.


For instance, we will consider the linear potential-outcome model in some examples:
\begin{align}
\label{potential outcome}
    y_i(\mathbf{Z}) =  \alpha + \psi(z_i, e_i) + \epsilon_i = \alpha + \beta_i z_i + \gamma_i e_i + \epsilon_i.
\end{align}


We are interested in estimating the average treatment effect (ATE)
\begin{eqnarray}\label{ate}
    \tau &=&\frac{1}{n}\sum_{i=1}^n Y_i(1,1) - \frac{1}{n}\sum_{i=1}^n Y_i(0,0) \\&=& \frac{1}{n}\sum_{i=1}^n (\alpha_i + \beta_i + \gamma_i + \epsilon_i - \alpha_i - \epsilon_i)\\ &=& \frac{1}{n}\sum_{i=1}^n (\beta_i + \gamma_i),
\end{eqnarray}
though we could, more generally, take any difference between exposure categories. 

We will consider the Horvitz-Thompson estimator for a given exposure threshold $h$:
\begin{eqnarray} \label{HT_estimator}
    \hat{\tau}_h &=& \frac{1}{n} \sum_{i=1}^n \frac{\mathbf{1}\{z_i = 1, e_i \geq h \}}{\mathbb{P}\{ z_i = 1, e_i \geq h\}} Y_i\\ && 
    - \frac{1}{n} \sum_{i=1}^n \frac{\mathbf{1}\{z_i = 0, e_i \leq 1-h \}}{\mathbb{P}\{ z_i = 0, e_i \leq 1- h\}} Y_i
\end{eqnarray}

Our goal then is to select the threshold $h$ in the Horvitz-Thompson estimator that minimizes the MSE. More formally, given a candidate set of thresholds $H$, we select
\begin{equation}
    h^* := \argmin_{h \in H} \text{MSE}(\hat{\tau}_h) = \argmin_{h \in H} \text{Bias}(\hat{\tau}_h)^2 + \text{Var}(\hat{\tau}_h).
\end{equation} 

The intuition here is that for stronger interference, i.e. when $f(e_i)$ is larger for a fixed $e_i$, the bias introduced by each edge connecting to a neighbor with a different treatment status is higher. The MSE-optimal Horvitz-Thompson estimator in that setting is expected to be one that incorporates a higher exposure threshold. We illustrate this through a display of a bias-variance trade-off with data from a linear model $Y_i = \alpha + \beta z_i + \gamma e_i + \epsilon_i$ across different thresholds for our adaptive Horvitz-Thompson estimator \ref{HT_estimator}, which we will call AdaThresh, in the top panel of Figure \ref{fig:bias_var_tradeoff}. This trade-off is illustrated for various ratios of $\gamma/\beta$. In the bottom panel of Figure \ref{fig:bias_var_tradeoff}, we illustrate how our approach detects the bias via linear regression for the different $\gamma/\beta$ ratios.
Using these bias estimates along with the exposure distribution, the MSE-optimal threshold, represented by the shaded regions, is selected (matching the top panel).


\begin{figure}
    \centering
    \begin{subfigure}[b]{0.49\linewidth}
     \centering     
     \includegraphics[width=\linewidth]{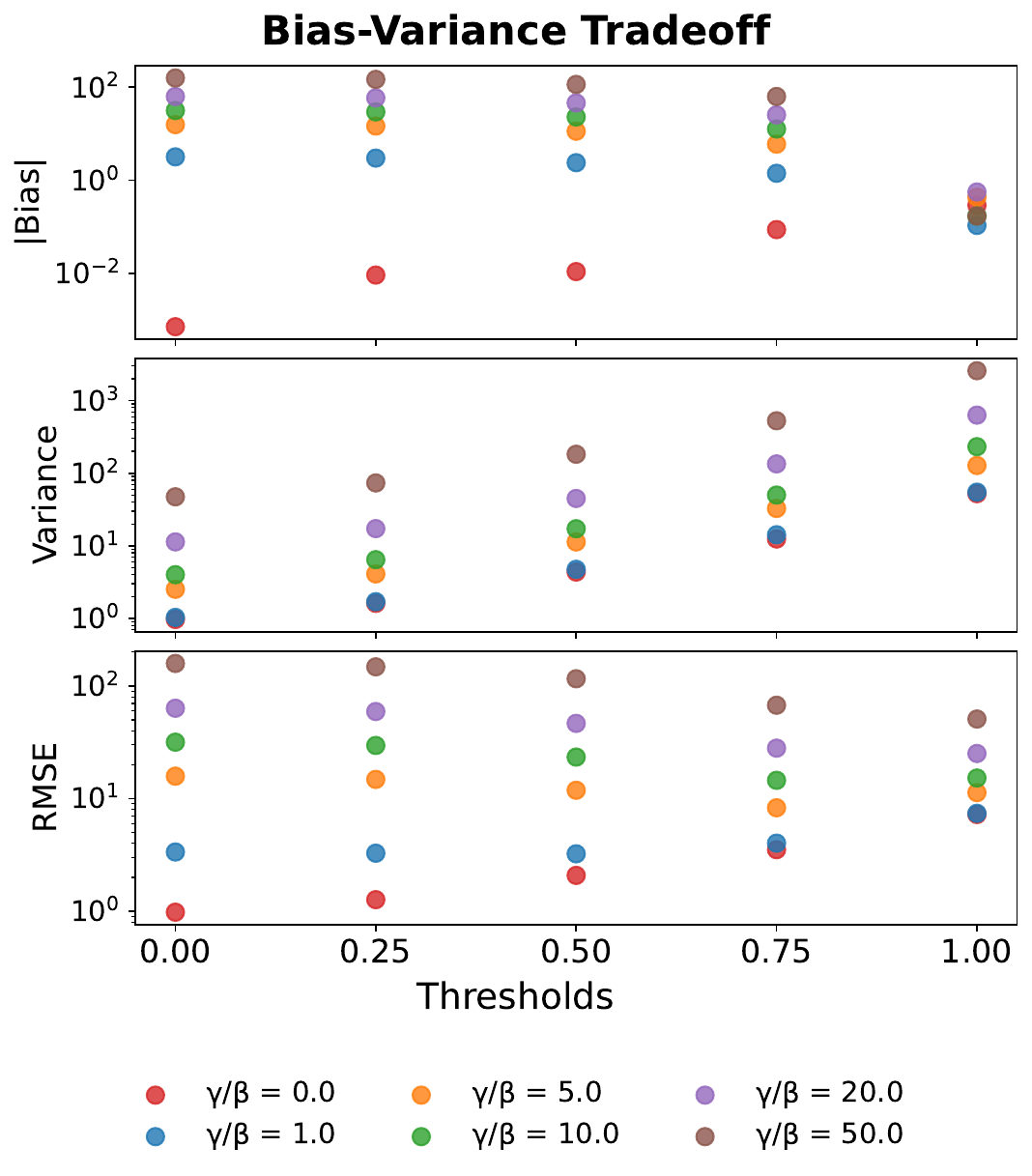}
     \end{subfigure}
    \hfill
     \begin{subfigure}[b]{0.49\linewidth}
    \centering
     \includegraphics[width=\linewidth]{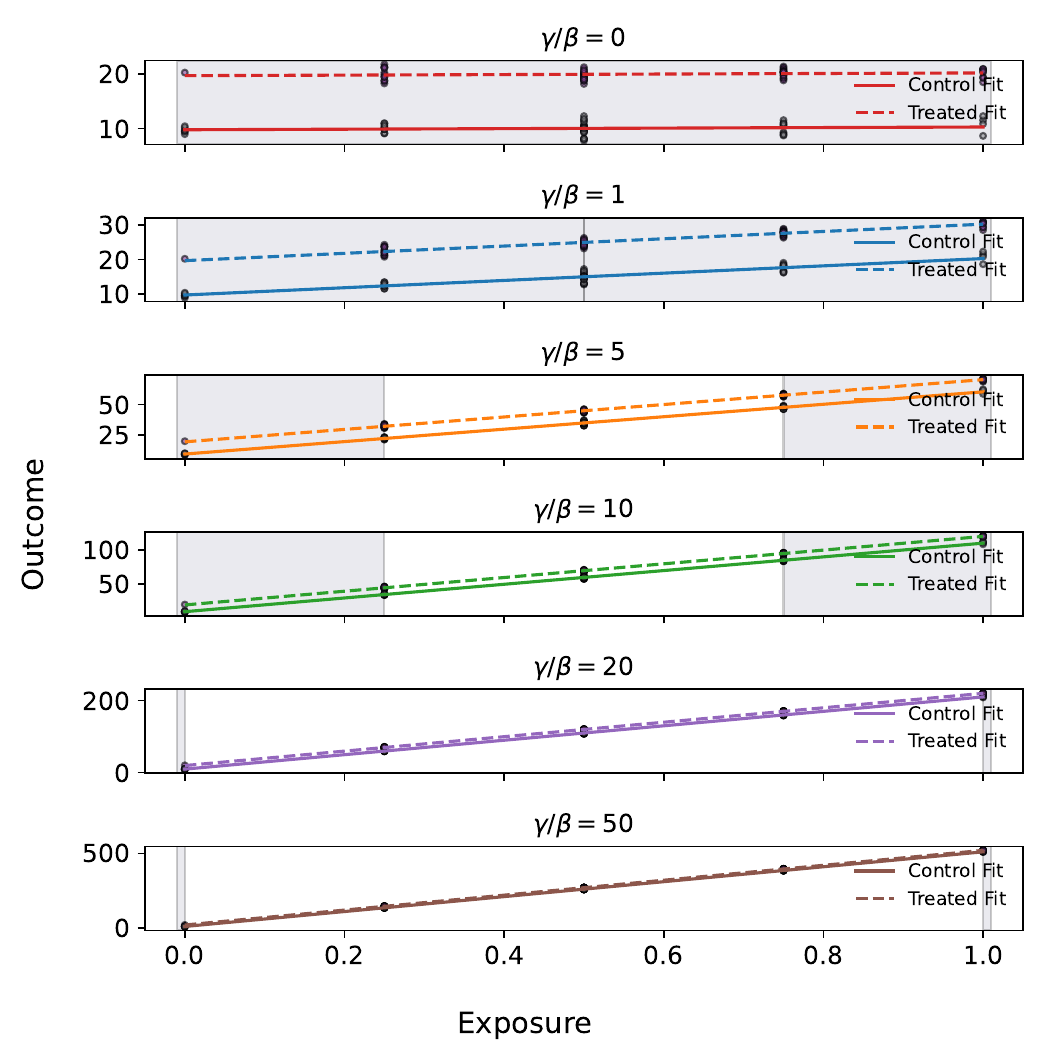}
     \end{subfigure}
     \caption{Left: bias, variance, and RMSE of AdaThresh across different thresholds for the 1000-node 2nd-power cycle graph with unit-level randomization, and linear model outcome $Y_i = 10 + 10 z_i + \gamma e_i +\epsilon_i$, for fixed $\epsilon_i$ generated from $\mathcal{N}(0,1).$  Right: toy illustration of linear fits used to approximate exposure bias caused by including more data across different thresholds. The ``jittered" datapoints signify the amount of variance reduction across thresholds. Rectangular blocks depict the data regions under the MSE-optimal thresholds.}
     \label{fig:bias_var_tradeoff}
\end{figure}

We take the MSE as the objective as we are interested in  precisely and accurately estimating the ATE, trading off between the bias and the variance of the estimator. See \citep[Section 2.2]{deng2024metric} for more discussions on this. 

We note that the tradeoff is not necessarily monotonic, as the bias and variance do not necessarily monotonically decrease and increase, respectively, due to the dependence of the units. However, the overall trend remains decreasing for bias and increasing for variance.

To illustrate how the bias and the variance of the estimator change across different thresholds, we draw upon toy examples from \citep{ugander2013graph} in the following subsection. Figure \ref{fig:Circulant_graphs} shows the different circulant graphs we consider under unit-level Bernoulli randomization and cluster randomization.  We extend this to general graphs in the subsequent sections.

\subsection{Toy examples: Tradeoffs in Circulant graphs} \label{circulant}

We first consider unit-level Bernoulli($p$) randomization in the $k$th-power cycle graphs. We say a graph is a $k$th power-Cycle graph if there exists an edge between each node and $2k$ of its nearest neighbors \citep{ugander2013graph}. We discuss more toy examples in Appendix \ref{circulant_more}. In the next section, we present results for general graphs.

\begin{proposition}[Absolute bias in k-th power cycle graphs under unit-randomization]
\label{bias_unit}
    When $p = 1/2$ and the potential outcome model is simply linear, i.e. $Y_i = \alpha + \beta z_i + \gamma e_i$, the absolute bias of the Horvitz-Thompson estimator for a given threshold $h \equiv l/2k$ for $l=0,2,...,2k$, in the $k$th-power cycle graph, under unit-level randomization, is equal to $\gamma \times \left[ \frac{\sum_{r=l}^{2k}\left( r/k - 1\right)\binom{2k}{r}}{\sum_{r=l}^{2k} \binom{2k}{r}} - 1\right].$
\end{proposition}

\begin{proposition}[Variance in k-th power cycle graphs under unit-randomization]
\label{var_unit}
    Denote the degree of the nodes by $d (=2k)$. When $p = 1/2$ and the potential outcome model is simply linear, i.e. $Y_i = \alpha + \beta z_i + \gamma e_i$, the variance of the Horvitz-Thompson estimator for a given threshold $h$, in the $k$-th power cycle graph, under unit-level randomization is proportional to 
\begin{align*}
\mathcal{O} \Bigg(\frac{1}{np^{dh}} & \bigg[(\alpha + \beta + \gamma dh )^2 + (\alpha + \beta + \gamma d(1-h))^2 \\
& - 2 \gamma h(1-h)d \bigg]\Bigg).
\end{align*}
\end{proposition}

\begin{figure}
     \centering
     \begin{subfigure}[b]{0.3\linewidth}
         \centering
         \includegraphics[height=0.8\textwidth, width=\textwidth]{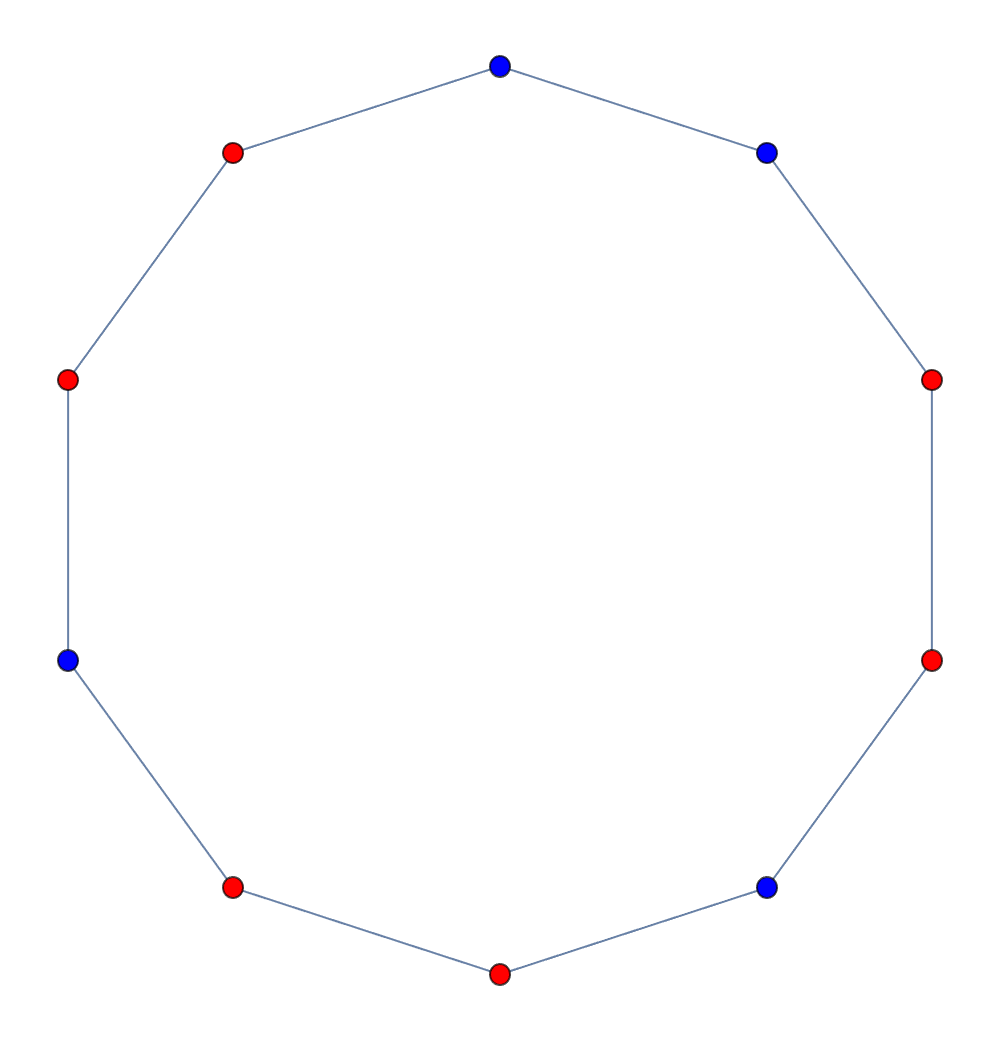}
         \label{fig:cycle_graph}
     \end{subfigure}
    \hfill
     \begin{subfigure}[b]{0.3\textwidth}
         \centering
         \includegraphics[height=0.8\textwidth, width=\textwidth]{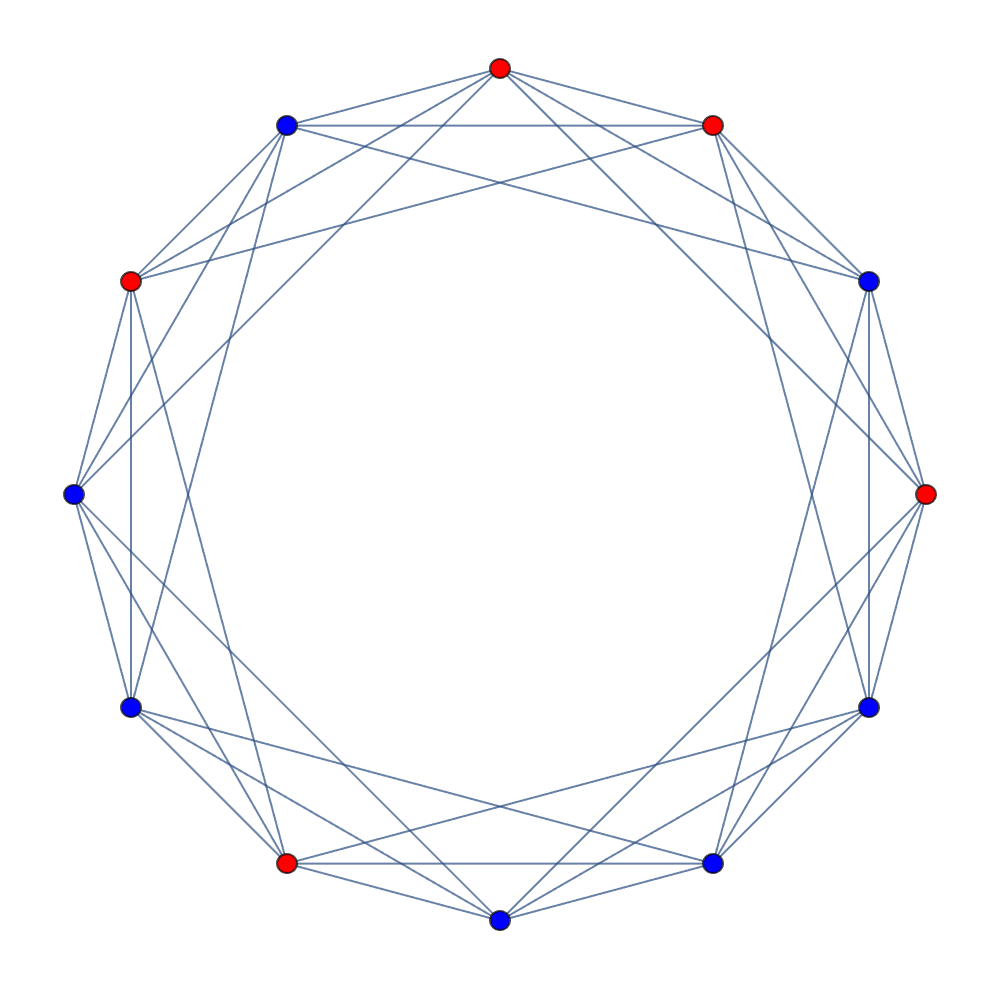}
         \label{fig:3rd_power_cycle}
     \end{subfigure}
     \hfill
     \begin{subfigure}[b]{0.3\textwidth}
         \centering
         \includegraphics[height=0.8\textwidth, width=\textwidth]{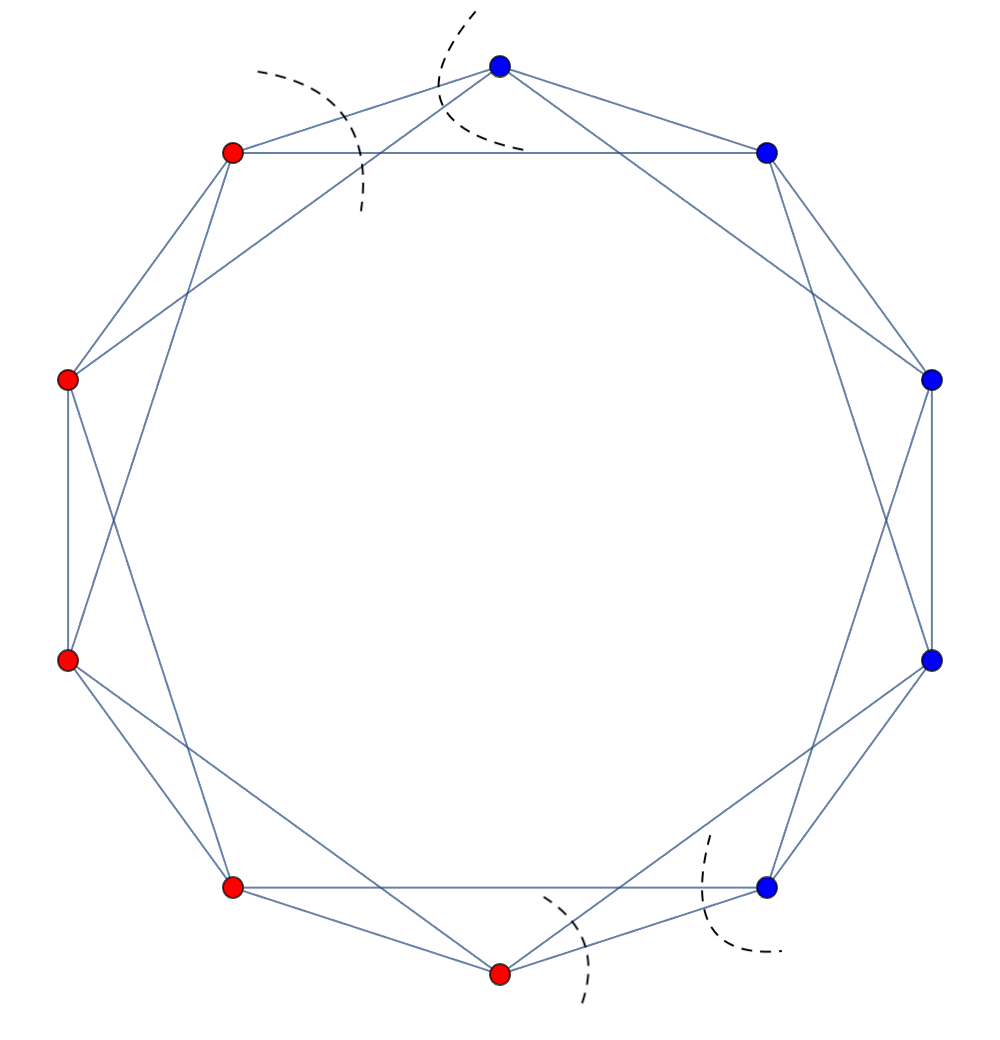}
         \label{fig:2nd_power_cycle}
     \end{subfigure}
     \caption{Circulant graphs with unit-level randomization and cluster-level randomization. Blue and red nodes represent treated and control units, respectively. Left: (1st-power) cycle graph with Ber(0.5) unit  randomization. Center: 3rd-power cycle graph with Benoulli(0.5) unit randomization. Right: 2nd-power cycle graph with Bernoulli(0.5) cluster randomization, with clusters of size 5 (=2k + 1).}
\label{fig:Circulant_graphs}
\end{figure}
Therefore, the optimal threshold $h^*$ here depends on $\gamma$, $\beta$, $n$, $p$, $d$, the graph structure, and the number of other candidate thresholds. 

From this toy example, it is not difficult to see that for unit-level Bernoulli randomization for more general graphs, the squared-bias would be $\Theta(\gamma^2h^2)$, and the variance would be $\Theta(p^{-dh}\beta^2/n).$ Therefore, minimizing the MSE corresponds to balancing these quantities.

\section{Estimating the rate of change of bias and variance}
In reality, however, one does not have access to the bias and variance. With the goal of thresholding in mind, we obtain a ``bias-signal" computed via a linear model, to estimate how the bias incurred changes by including more data under a wider bandwidth corresponding to a smaller threshold $h$. We propose the following absolute bias-signal estimator:
 \begin{align*}
     \hat{b}(\hat{\tau}_h) &=& \frac{1}{n} \sum_{i=1}^n \frac{(1-e_i) \hat{\gamma}_n\mathbf{1}\{ Z_i = 1, e_i \geq h\}}{\mathbb{P}\{ Z_i = 1, e_i \geq h\}}\\ &&+ \frac{1}{n} \sum_{i=1}^n \frac{e_i \hat{\gamma}_n\mathbf{1}\{ Z_i = 0, e_i \leq  1- h\}}{\mathbb{P}\{ Z_i = 0, e_i \leq 1- h\}},
 \end{align*} where $\hat{\gamma}_n$ is the linear regression coefficient of the outcome on the exposure variable.

If the potential outcomes depend strongly and postively on the treatment exposures, $\hat{\gamma}_n$ will be larger. Consequently, our estimator will capture a more pronounced increase in bias caused by including more data when using a smaller threshold $h$.

Next, to estimate the reduction in the variance for a drop in the threshold $h$, we estimate the variance of the Horvitz-Thompson estimator at a given threshold $h$ using the variance estimator discussed in \citep{aronow2017estimating}. The variance estimator is available in Appendix \ref{appendix:variance_estimator}

We run linear regression purely to obtain a signal of the bias, which we then use to pick an optimal threshold. That is, we are not assuming an underlying linear model, in which case running linear regression on the covariates ($z_i, e_i$) and obtaining the coefficient $\hat{\beta}$ of $z_i$ would be the optimal estimate of the ATE. 

One could split the sample to run linear regression (or a more complex learning mechanism) on one part of the sample and then use the bias signal obtained there to pick the optimal threshold on the other part of the sample. Since the units of interest are dependent (as modeled by the network interference), one would have to prove performance guarantees by leveraging results such as \cite{hart1990data}.  In our setting, however, it is not necessary that we split the data as we leverage the simplicity of the model class. In Appendix \ref{double_dipping}, we discuss further why our methodology is sound without sample-splitting by providing a proof sketch on how Donsker results hold in our setting, to show that we are not overfitting.

\section{Theoretical and Empirical Results}

We now present theoretical and simulation results in the context of general graphs. 
 Denote by $H$ the set of exposure thresholds we consider. For example, in the cycle graphs, $H = \{ \frac{u}{d}: u \in \{ 0, 1,..., d\}\}$ where $d$ is the degree of the graph.

\subsection{Theory}

\begin{assumption}[Positivity]
       \label{cond:positivity}
    For all $i \in [n]$, $r \in \{ 0, 1\}$,
    \begin{align*}
        \mathbb{P}(Z_i = r) > 0
    \end{align*} 
\end{assumption}

This assumption ensures that we have sufficient data to estimate the ATE. The positivity assumption \ref{cond:positivity} also relates to statistical leverage \citep{mahoney2011randomized, martinsson2020randomized,young2019channeling, pilanci2015randomized}. The idea, informally, is that under positivity, statistical leverage across the input data is more homogeneous. As a result of this, there are no observations that have too much control over the linear regression. 


\begin{assumption}[Bounded variation treatment-exposure function]
\label{cond:bounded_var}
Let the potential outcome function be $y_i = \alpha + \psi(z_i, e_i) + \epsilon_i$ as in (\ref{potential outcome}). The function
$\psi(z_i, e_i)$ has bounded variation.
\end{assumption}




\begin{assumption}[Bounded first-order interactions]
\label{cond:bounded_first}
Let $B_r(i)$ denote the ball of radius $r$ around unit $i$.
For all $i \in [n]$,
    \[ |B_1(i)| = O(1).\] 
\end{assumption}

The bounded-degree condition above is important under unit-randomization design. For a cluster-randomization design, we would require bounded interactions across clusters. In particular the graph interference structure should have small $k_n$-conductance (i.e. the $k_n$-partition generalization of the Cheeger constant) for $k_n$ growing with $n$. Though restrictive theoretically, it is consistent with the observation that networks in practice tend to be sparse and clustered \citep{chandrasekhar2020testing}. Violating this condition will violate exposure positivity (see remark below), and Proposition \ref{prop:stochastic_equi} will not hold.

Some examples of graph classes that fall under the categorization of Assumption \ref{cond:bounded_first} are expander graphs, and growth-restricted graphs \citep{alon1986eigenvalues, arora2009expander, gkantsidis2003conductance, kuhn2005locality, krauthgamer2003intrinsic, kowalski2019introduction}. In the causal inference literature, the effect of the growth-restricted graph setting was studied in \citep{ugander2013graph}. 

\begin{remark}In the weighted graph setting, one could also consider equivalence classes $\Tilde{e}_i$ of exposures. That is, suppose $\Tilde{e}_i \equiv \Tilde{e}_j$ if and only if $|e_i - e_j|<\varepsilon_n$ for $\varepsilon_n$ specified below, and that exposure positivity, i.e. for all $i \in [n]$, $r \in \{0,1\}$, and $s \in [0,1]$, $\mathbb{P}(Z_i = r, \Tilde{e}_i = s) > 0$ is satisfied. Then, we can weaken our bounded degree constraint. Let $B_r(i)$ denote the ball of radius $r$ around unit $i$, such that it decomposes, $|B_1(i)| = |B_1^S(i)| + |B_1^W(i)|$ for all $i \in [n]$. $B_1^S(i)$ is the ball of radius 1 around unit $i$ with strong connections $\eta_{is}$, and $|B_1^W(i)|$ be the ball of radius 1 around unit $i$ with weak connections $\delta_{iw}$ such that $\sup_{i,s,w} \frac{\eta_{is}}{\delta_{iw}} \to \infty$. For all $i \in [n]$, we require that $ |B_1^S(i)| = \mathcal{O}(1)$. With an abuse of notation, write $e_i(\mathcal{N})$ to mean the exposure of unit $i$ as a function of the treated units in the subset $\mathcal{N} \subseteq B_1(i)$. Then, we require also also that $e_i(B_1^S(i)) = \Omega(e_i(B_1(i)))$. Additionally, defining $v_i = z_i - \mathbf{E}[z_i]$ for random variables $z_i$, $|B_1(i)|$ must satisfy the three conditions in \citep{chandrasekhar2023general} generalizing \citep{aronow2017estimating}, with an additional condition that the sum of the affinity set covariances is $\Omega(n)$, for the bias terms as well as the variance terms in the MSE. Since those conditions nest strong mixing dependence structure, which nests our current results for the variance terms, these conditions are sufficient. Then, taking $\varepsilon_n := \varepsilon(\delta_n, B_1^W(i)) > 0$, gives us uniform control on the estimated MSE of the estimator with respect to the true MSE of the estimator, and all the results follow through. 
\end{remark}

In the following theorem, we characterize the probability of choosing the correct threshold under the best average linear fit. That is, define $\gamma^* := \argmin_{\gamma \in \mathbb{R}} \mathbb{E}[Y - \gamma e].$ The correct threshold $h^*$ under the best average linear fit is the threshold minimizing the sum of the true variance and the true squared-bias under $Y_i = \alpha + \beta z_i + \gamma^* e_i.$ We write $M_n^*(h), \hat{M}_n(h)$ to denote the MSE under the true average best fit line and the MSE estimated by our methods, respectively, at threshold $h$ and finite-population size $n$.

\begin{theorem} \label{main_thm} Suppose Assumptions \ref{cond:positivity}, \ref{cond:bounded_var}, and \ref{cond:bounded_first} hold. Let $\Delta_n := \min_{h,h' \in H} |M_n^*(h) - M_n^*(h') |.$ Let $U_n \in \mathbb{R}$ be such that $\max_h |b_n^*(\hat{\tau}_h)| \leq U_n$ (such a $U_n$ exists by Conditions \ref{cond:positivity}, \ref{cond:bounded_var}), and let $h_n^*$ denote the optimal threshold under the true average best fit line, while $\hat{h}_n$ denotes the threshold chosen by our methods. Denote also the treatment assignment probability under unit-level Bernoulli randomization by $p$, and D is such that $D \geq |H|$, where $H$ the set of exposures. Then, 
\begin{align*}
    &\mathbb{P}\{\hat{h}_n \neq h_n^* \} \\
    &\leq \sum_{h \in H} \mathbb{P}\{| \hat{M}_n(h) - M_n^*(h) | > \Delta_n / 2\} \\
    &\leq 3D \max \bigg\{
        \exp\bigg( -\frac{\Delta_n n p(1-p)}{8 c d_{\max}} + 1 \bigg), \\
    &\quad \exp\bigg( -\frac{ \Delta_n^2 n p(1-p)}{(16U_n)^2 cd_{\max}} + 1 \bigg), \\
    &\quad 6\exp\bigg( - \frac{C n \big(\frac{\Delta_n}{4} - \frac{c d_{max}^2 }{n}\big)}{
        \sqrt{A_{n,2}} + \sqrt{\big(\frac{\Delta_n}{4} - \frac{c d_{max}^2}{n} \big)M_{n,2}}
    }\bigg)
    \bigg\},
\end{align*} for some constants $c, C$, and where $A_{n,p} \leq 16 \| \Tilde{v} \|_{L_1}^2\{c_1 + \frac{(\log n)^4}{n}\}^2$, for some constant $c_1$ and $\Tilde{v}$ is the Fourier transform of the summands of the variance components, and $M_{n,2} = 4 \| \Tilde{v} \|_{L_1} (\log n)^4.$ 

If, instead, the design is cluster-level Bernoulli randomization with probability $p$, denote the maximum number (plus one) of edges between clusters by $s_{\max}.$  Then,
\begin{align*}
    &\mathbb{P}\{\hat{h}_n \neq h_n^* \} \\
    &\leq \sum_{h \in H} \mathbb{P}\{| \hat{M}_n(h) - M_n^*(h) | > \Delta_n / 2\} \\
    &\leq 3D \max \bigg\{
        \exp\bigg( -\frac{\Delta_n n p(1-p)}{8 c s_{\max}} + 1 \bigg), \\
    &\quad \exp\bigg( -\frac{ \Delta_n^2 n p(1-p)}{(16U_n)^2 c s_{\max}} + 1 \bigg), \\
    &\quad 6\exp\bigg( - \frac{C n \big(\frac{\Delta_n}{4} - \frac{c s_{max}^2 }{n}\big)}{
        \sqrt{A_{n,2}} + \sqrt{\big(\frac{\Delta_n}{4} - \frac{c s_{max}^2}{n} \big)M_{n,2}}
    }\bigg)
    \bigg\},
\end{align*} for some constants $c, C$, and where $A_{n,2} \leq 16 \| \Tilde{v} \|_{L_1}^2\{c_1 + \frac{(\log n)^4}{n}\}^2$, for some constant $c_1$ and $\Tilde{v}$ is the Fourier transform of the summands of the variance components, and $M_{n,2} = 4 \| \Tilde{v} \|_{L_1} (\log n)^4.$ 
\end{theorem}


In Theorem \ref{main_thm}, we note that the MSE gap is lower-bounded by a constant, i.e. $\Delta_n \geq \Delta \equiv  \min_{h,h' \in H;n} |b_n^*(h) - b_n^*(h') |$. The proof to the theorem is in Appendix \ref{appendix:proof_main}. In this proof, we make use of the results from \citep[Theorem 3.1]{ziemann2024noise}, and \citep[Theorem 2.1]{shen2020exponential}.

This gives us the following corollary. Denote the true MSE by $M_n^{**}$, and the corresponding bias and optimal threshold by $b_n^{**}$ and $h_n^{**}$, respectively. Let $U_n^* \in \mathbb{R}$ be such that $\max_h |b_n^{**}(\hat{\tau}_h)| \leq U_n^*$ (such a $U_n^*$ exists by Conditions \ref{cond:positivity}, \ref{cond:bounded_var}). We note again that $U_n^*$ is bounded from below.

\begin{corollary} \label{corollary:approx_error}
    Suppose that $\sup_{e_i}| f(e_i) - \gamma^* e_i| \leq \delta$. Let $\Tilde{\delta}:= 16 \delta^2  + 8 \delta U_n^*$ . Then,
    \begin{align*}
        &\mathbb{P}\{ \hat{h}_n \neq  \hat{h}_n^{**} \} \\
        &\leq \sum_{h \in H} \mathbb{P}\{ |\hat{M}_n(h) -  M_n^{**}(h)| > \Delta/2 \}\\
        &\leq 3D \max \bigg\{
        \exp\bigg( -\frac{(\Delta/8 - \Tilde{\delta}) n p(1-p)}{c d_{\max}} + 1 \bigg), \\
    &\quad \exp\bigg( -\frac{ (\Delta^2/(16U_n)^2 - \Tilde{\delta}) n p(1-p)}{ cd_{\max}} + 1 \bigg), \\
    &\quad 6\exp\bigg( - \frac{C n \big(\frac{\Delta_n}{4} - \frac{c d_{max}^2 }{n} - \Tilde{\delta} \big)}{
        \sqrt{A_{n,2}} + \sqrt{\big(\frac{\Delta_n}{4} - \frac{c d_{max}^2 }{n} - \Tilde{\delta} \big)M_{n,2}}
    }\bigg)
    \bigg\},
    \end{align*} for some constants $c, C$, and where $A_{n,p} \leq 16 \| \Tilde{v} \|_{L_1}^2\{c_1 + \frac{(\log n)^4}{n}\}^2$, for some constant $c_1$ and $\Tilde{v}$ is the Fourier transform of the summands of the variance components, and $M_{n,2} = 4 \| \Tilde{v} \|_{L_1} (\log n)^4.$ 

If, instead, the design is cluster-level Bernoulli randomization with probability $p$, denote the maximum number (plus one) of edges between clusters by $s_{\max}.$  Then, we can replace $d_{\max}$ by $s_{\max}$ in the bounds above.
\end{corollary}

Corollary \ref{corollary:approx_error} tells us that if the maximum linear approximation error between the best average linear fit and the true exposure function $f(\cdot)$ is small enough relative to the minimum MSE gap $\Delta$, our estimator will be optimal with high probability for large $n$. If $f(x) = \gamma x$ for all $x \in [0,1]$ indeed, then with necessarily we have that with high probability for large $n$, our estimator is optimal.

\subsection{Simulations}
\label{simulations}
We now compare the performance of our approach to that of the Horvitz-Thompson estimator with threshold one, and the Horvitz-Thompson estimator with threshold zero. Additionally, since our objective is essentially an MSE-optimal bandwidth selection problem for the indicator kernel (see Appendix \ref{appendix:bandwidth}, \ref{HT_estimator_bandwidth}), we compare our estimator to the Horvitz-Thompson estimator with the plugin threshold from Lepski's method \citep{goldenshluger2011bandwidth, su2020adaptive}. We note that Lepski's method requires monotonicity of the bias and the variance to work well. In our setting, due to the implicit dependency structure from the graph in the inverse-propensity weights, the monotonicity assumption may be violated (see top panel of Figure \ref{fig:bias_var_tradeoff}). We write out the three estimators in Appendix \ref{appendix:HT_comparison_estimators}.

In Figure \ref{fig:HT_GLR_sims}, we simulate outcomes from the linear model with $\alpha =10$, $g(z_i) = \beta z_i = 10z_i$, $f(e_i) =\gamma e_i$, and fixed $\epsilon_i$ generated from $\mathcal{N}(0,1)$ for a 1000-node 2nd-power cycle graph. Figure \ref{fig:HT_sbm_sims} displays the simulation results for a 200-node SBM graph (see Figure \ref{fig:sbm}). We focus on varying the ratio $\gamma/\beta$ as we consider a fixed graph. We see that our adaptive thresholding estimator (AdaThresh) generally performs better than existing estimators, ``interpolating" between the fixed 0/1-threshold Horvitz-Thompson estimators, and out-performing the Lepski's method plug-in estimators when the threshold $\gamma/\beta$ is high.

\begin{figure}
    \centering
    \begin{subfigure}[b]{\linewidth}
     \centering     
     \includegraphics[width=\linewidth]{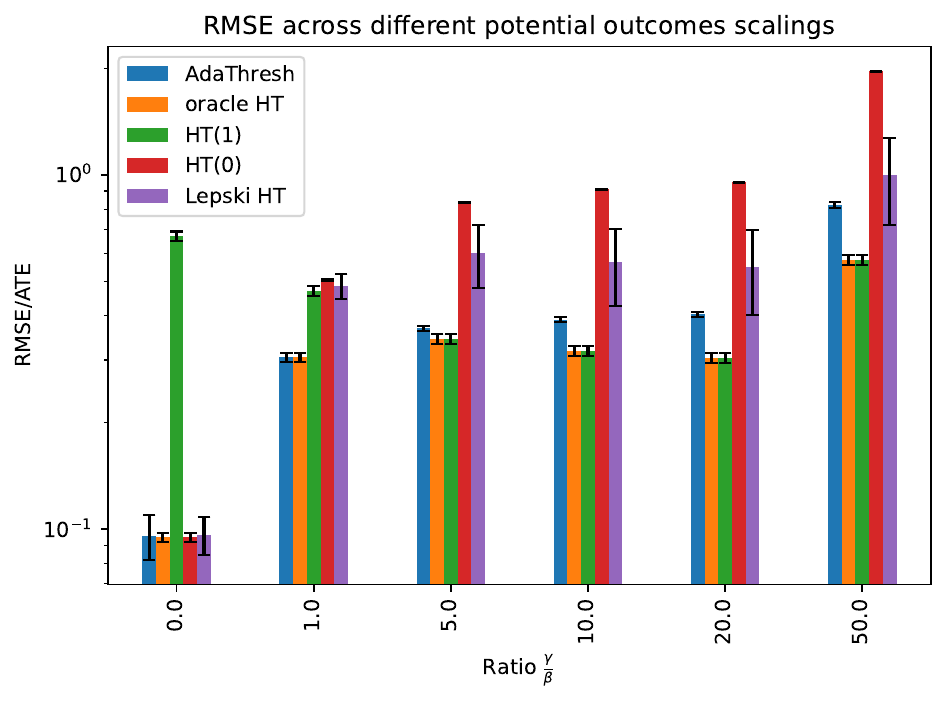}
     \end{subfigure}
    \hfill
     \begin{subfigure}[b]{\linewidth}
    \centering
     \includegraphics[width=\linewidth]{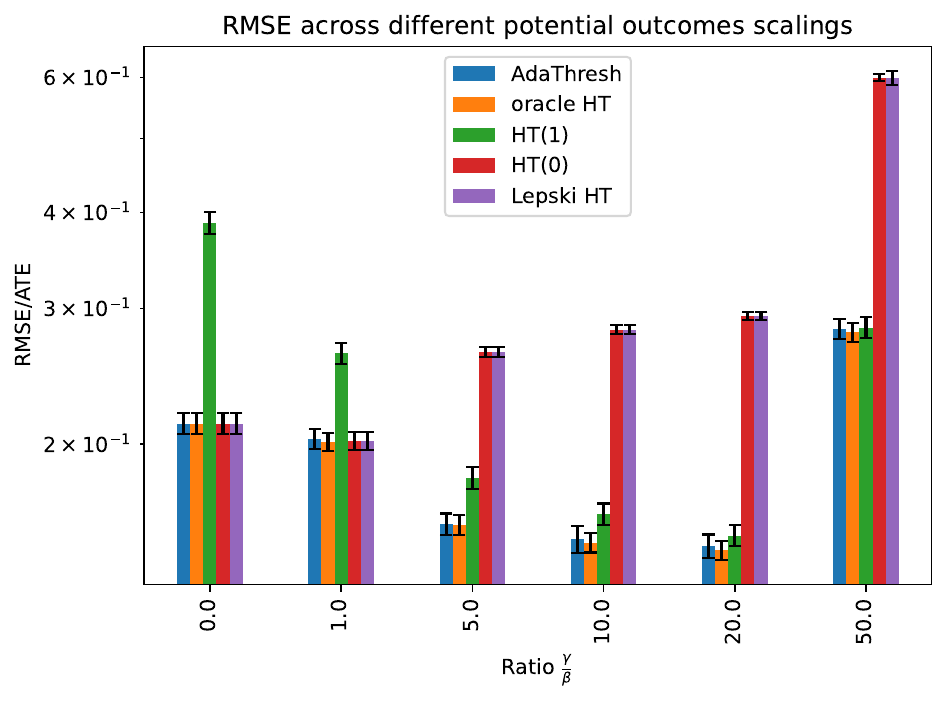}
     \end{subfigure}
\caption{RMSE (normalized by the ATE) of different Horvitz-Thompson estimators. Top: 2nd-power cycle graph under unit-level Ber(0.5) randomization. Bottom: 2nd-power cycle graph under cluster-level Ber(0.5) randomization with cluster sizes 5 (=2k + 1). The error bars are two times the standard deviation.}
\label{fig:HT_GLR_sims}
\end{figure}

\begin{figure}
    \centering
\includegraphics[width=\linewidth, height=5cm, keepaspectratio]{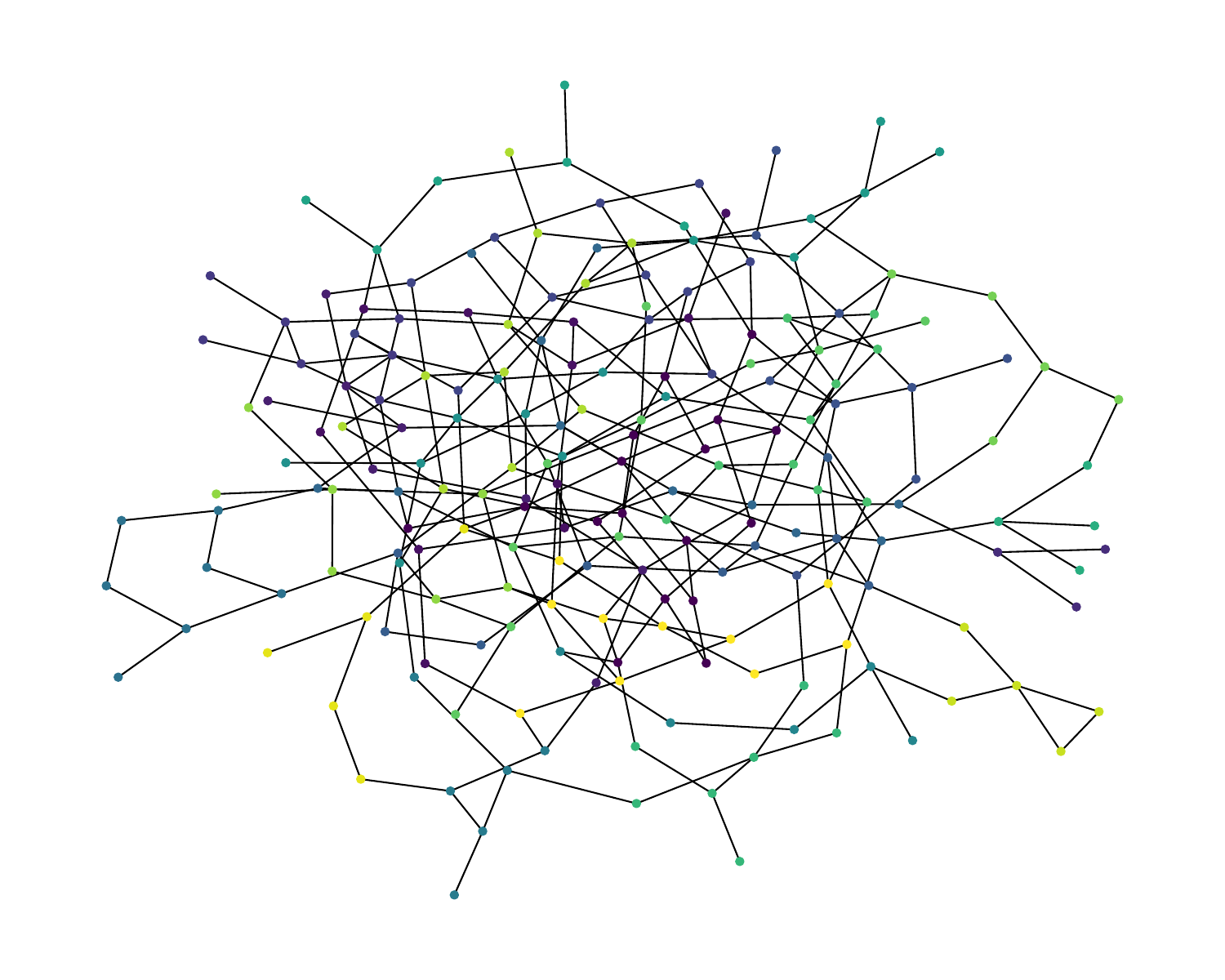}
    \caption{A graph of size $n=200$ generated from the Stochastic Block Model (SBM), under unit-level and cluster-level Bernoulli (0.5) randomization. Node colors reflect cluster membership by K-means (with ground truth $K = 25$) for cluster-randomization. }
    \label{fig:sbm}
\end{figure}

\begin{figure}
    \centering
    \begin{subfigure}[b]{\linewidth}
     \centering     
     \includegraphics[width=\linewidth]{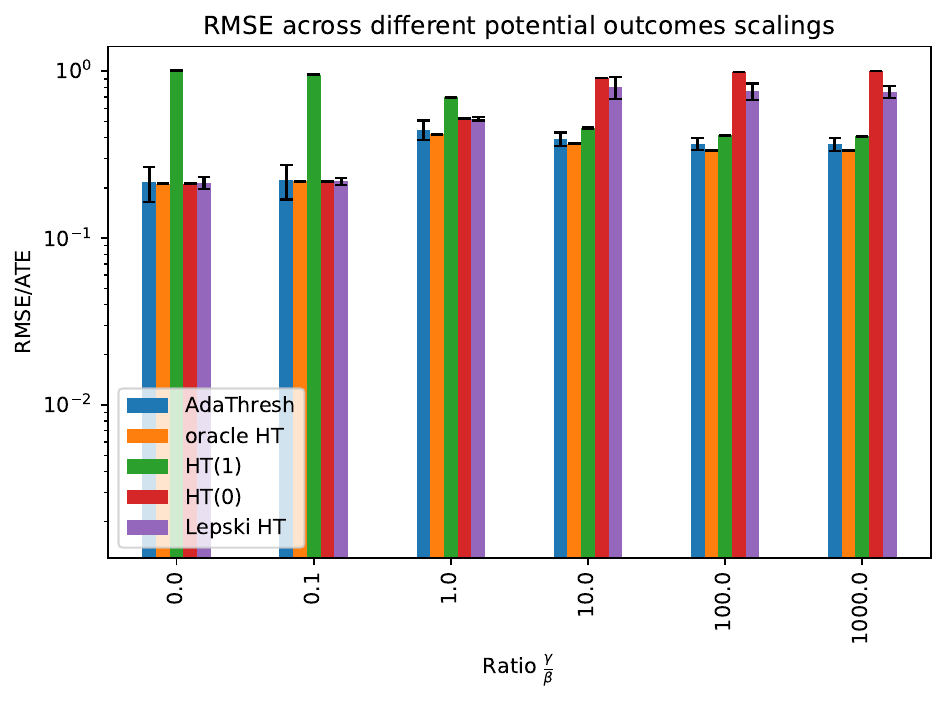}
     \end{subfigure}
    \hfill
    \centering
    \begin{subfigure}[b]{\linewidth}
\includegraphics[width=\linewidth]{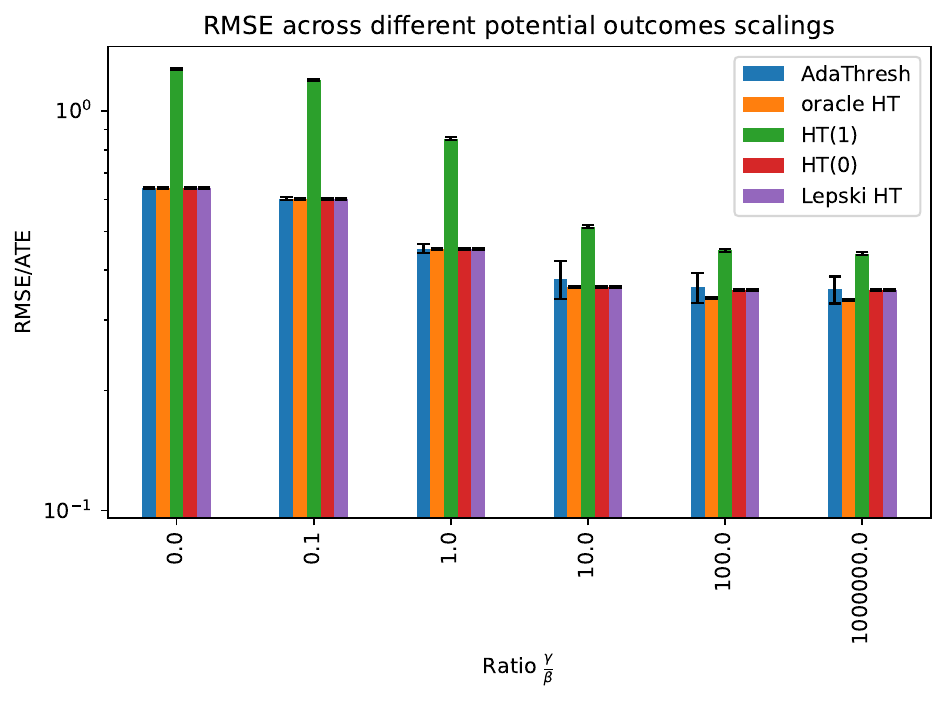}
     \end{subfigure}
\caption{RMSE (normalized by the ATE) of different Horvitz-Thompson estimators for the SBM graph in Figure \ref{fig:sbm}. Top: unit-level Ber(0.5) randomization. Bottom: cluster-level Ber(0.5) randomization with 25 clusters. The error bars are two times the standard deviation.}
\label{fig:HT_sbm_sims}
\end{figure}

\subsection{Real Data} \label{real_data}
We evaluate the performance of our estimator on the Amazon (DVD) products similarity network \citep{leskovec2007dynamics}. We consider the performance of the estimator on a subset of (the first) 1000 non-isolated nodes from a total of 19828 nodes, with exposure probabilities computed on the full graph. We generate simulated data using the linear model with $\alpha = 10, g(z_i) = \beta z_i = 10z_i$, $f(e_i) =\gamma e_i$  
, and $\epsilon_i$ is generated from $\mathcal{N}(0,1)$ under a unit-level Bernoulli(0.5) randomization design. We focus on varying the ratio $\gamma/\beta$ as we consider a fixed graph. Figure \ref{fig:real_data} shows that we consistently perform better than existing fixed estimators. See Appendix \ref{appendix:simulations_technical} for more details.

\begin{figure}
    \centering
\includegraphics[width=\linewidth]{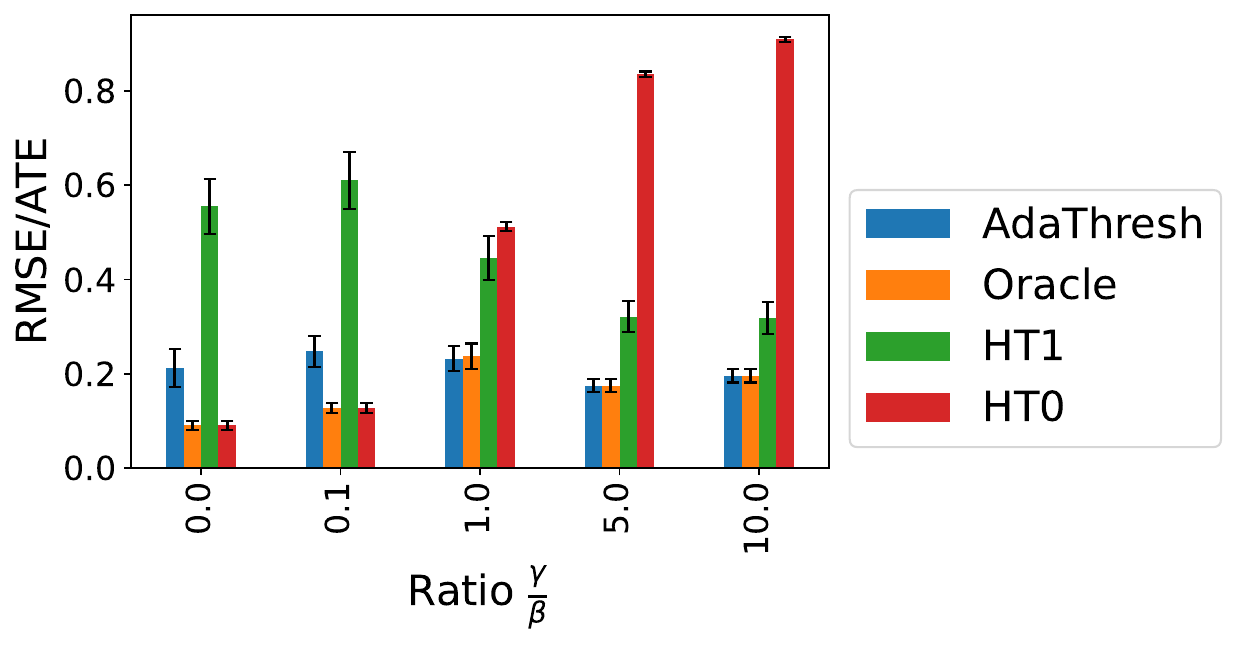}
    \caption{RMSE (normalized by the ATE) across different thresholds on the Amazon product network data. The error bars are two times the standard deviation.}
    \label{fig:real_data}
\end{figure}

\section{Discussion}
In this paper, we investigated the adaptive Horvitz-Thopson estimator for symmetric thresholds $h$ and $1-h$ for treatment and control, respectively. One could also use this approach for a more general estimator with thresholds $h_1$ and $1-h_0$, respectively with $h_1 \neq h_0$. Additionally, in Appendix \ref{appendix:difference-in-means}, we demonstrate the performance of this approach using the Difference-in-Means estimator incorporating exposure thresholds. We leave it to future work to extend our results to more general exposure estimation procedures. In \citep{chin2019regression}, the author proposes learning the feature variables that are most predictive of the potential outcomes. One could extend our work by first learning the feature variables, as proposed by \citep{chin2019regression}, followed by then learning the appropriate optimal threshold associated with these features, using our approach.

In Appendix \ref{appendix:non-linear}, we illustrate the robustness of our estimator to deviations from linearity in the potential outcomes model. We note that our framework would apply in more general settings, including a direct extension to friends-of-friends (or further connections) or to neighborhoods that are either observed or learned through a clustering algorithm.  An additional interesting future direction could be to explore an extension of our work to the framework presented in~\cite{chandrasekhar2023general}, which allows each individual to have a unique exposure and then bins exposures that are separated by at least a threshold.

Furthermore, to improve upon robustness to non-linear settings, we propose an extension using local regression to estimate the rate of change of bias within the $[0,1-h],[h,1]$ windows. As long as there is sufficient concentration of exposures in these windows, robustness is never worse.
In Appendix \ref{appendix:localLR}, we display our simulation results in this setting using local linear regression to estimate the rate of change of bias within the $[0,1-h],[h,1]$ windows. One could also use other local regression approaches, such as kernel regression, etc., while maintaining sufficiently small model complexity to avoid needing to sample-split.

This paper prioritizes minimizing the mean squared error (MSE) to effectively balance the bias-variance trade-off in a Horvitz-Thompson-type estimator, rather than focusing on inference. However, we acknowledge the significance of characterizing inferential tools, such as Wald-type confidence intervals, in specific contexts. Notably, \citep{aronow2017estimating} and \citep{athey2018exact} offer valuable insights into inferential methods for estimating average treatment effects under network interference.



\small
\bibliographystyle{unsrtnat}
\bibliography{ref}

\medskip

\appendix

\onecolumn
\section{Appendix / supplemental material}

\subsection*{Notation}
In this supplement, we write $b^*(\tau_h),\hat{\tau_b}(h),v^*(\tau_h),\hat{v}(\tau_h))$ to represent the true bias, estimated bias, true variance, and estimated variance, respectively, at threshold $h$. We also write $d_max$ to represent the upper-bound, by \ref{cond:bounded_first}, on the degree of nodes in the network.

\subsection{Variance estimator}
\label{appendix:variance_estimator}

We use the variance estimator, for the Horvitz-Thompson estimator, proposed by \citep[Eqn. 7,9]{aronow2017estimating}. 
\begin{equation}
    \begin{split}
    \hat{v}(\hat{\tau}_h) &= \sum_{i=1}^n \frac{\mathbf{1}\{z_i = 1, e_i \geq h\}Y_i^2}{n^2\pi_i^1} (\frac{1}{\pi_i^1} - 1) \\
    &+ \sum_{i=1}^n \sum_{\mathclap{\substack{j=1\\ j\neq i}}}^n \frac{\mathbf{1}\{z_i = 1, e_i \geq h\}\mathbf{1}\{z_j = 1, e_j \geq h\}Y_iY_j}{n^2\pi_{ij}^{11}} (\frac{\pi_{ij}^{11}}{\pi_i^1 \pi_j^1} - 1) \\
    &+ \sum_{i=1}^n \frac{\mathbf{1}\{z_i = 0, e_i \leq 1-h\}Y_i^2}{n^2\pi_i^0} (\frac{1}{\pi_i^0} - 1)\\
    &+ \sum_{i=1}^n \sum_{\mathclap{\substack{j=1\\j\neq i}}}^n \frac{\mathbf{1}\{z_i = 0, e_i \leq 1-h\}\mathbf{1}\{z_j = 0, e_j \leq 1-h\}Y_iY_j}{n^2\pi_{ij}^{00}} (\frac{\pi_{ij}^{00}}{\pi_i^0 \pi_j^0} - 1) \\
    &- \frac{2}{n^2} \sum_{i=1}^n \sum_{j\in [n]; \pi_{ij}^{10} > 0} ( \mathbf{1}\{z_i = 1, e_i \geq h\}\mathbf{1}\{z_j = 0, e_j \leq 1-h\}Y_i Y_j)(\frac{1}{\pi_i^1 \pi_j^0} - \frac{1}{\pi_{ij}^{10}} ) \\
    &+ \frac{2}{n^2} \sum_{i=1}^n \sum_{j\in [n]; \pi_{ij}^{10} = 0} ( \frac{\mathbf{1}\{z_i = 1, e_i \geq h\}Y_i^2}{2 \pi_i^1} + \frac{\mathbf{1}\{ z_j = 0, e_j \leq 1-h \}Y_j^2}{2 \pi_j^0})
    \end{split}
\end{equation}

Here, $\pi_i^{x}, \pi_{ij}^{xy}$ are defined with respect to the threshold $h$. That is, 
$\pi_{i}^{1} = \mathbb{P}\{z_i =1 , e_i \geq h\}$, 
$\pi_{i}^{0} = \mathbb{P}\{z_i =0 , e_i \leq 1-h\}$,
$\pi_{ij}^{10} = \mathbb{P}\{z_i =1 , e_i \geq h, z_j = 0, e_i \leq 1-h \}$, 
$\pi_{ij}^{11} = \mathbb{P}\{z_i =1 , e_i \geq h, z_j = 1, e_i \geq h \}$,
$\pi_{ij}^{01} = \mathbb{P}\{z_i =0 , e_i \leq 1-h, z_j = 1, e_i \geq h \}$,
$\pi_{ij}^{00} = \mathbb{P}\{z_i =0 , e_i \leq 1-h, z_j = 0, e_i \leq 1-h \}$.

\subsection{Horvitz-Thompson estimators with existing approaches}\label{appendix:HT_comparison_estimators}

We write the vanilla Horvitz-Thompson estimator (at threshold one) \ref{est:HT}, the Horvitz-Thompson estimator at threshold zero \ref{est:HT0}, and the Horvitz-Thompson estimator with the threshold selected via Lepski's method \ref{est:Lepski_HT} below.

\subsubsection{Vanilla Horvitz-Thompson estimator (at threshold one)}
\begin{align} \label{est:HT}
        \hat{\tau}_{\text{HT}_1}  = \frac{1}{n} \sum_{i=1}^n \frac{\mathbf{1}\{ Z_i = 1, e_i = 1\}}{\mathbb{P}\{ Z_i = 1, e_i = 1\}} Y_i - \frac{1}{n} \sum_{i=1}^n \frac{\mathbf{1}\{ Z_i = 0, e_i = 0\}}{\mathbb{P}\{ Z_i = 0, e_i = 0\}} Y_i
\end{align}

\subsubsection{Horvitz-Thompson estimator at threshold zero}
    \begin{align} \label{est:HT0}
        \hat{\tau}_{\text{HT}_0}  = \frac{1}{n} \sum_{i=1}^n \frac{\mathbf{1}\{ Z_i = 1\}}{\mathbb{P}\{ Z_i = 1\}} Y_i - \frac{1}{n} \sum_{i=1}^n \frac{\mathbf{1}\{ Z_i = 0\}}{\mathbb{P}\{ Z_i = 0\}} Y_i
\end{align}

\subsubsection{Lepski's method}
As described in \citep{su2020adaptive}, we first take
\begin{align}
\label{lepski_interval}
    I(h) := [ \hat{\tau}_{\text{HT}_h} - 2 \hat{\text{SDEV}}(\hat{\tau}_{\text{HT}_h}), \hat{\tau}_{\text{HT}_h} + 2 \hat{\text{SDEV}}(\hat{\tau}_{\text{HT}_h})].
\end{align}
Then, take
\begin{align}
\label{lepski_h}
    \hat{h}_\text{Lepski} := \min \{h \in H : \cap_{h=1}^k I(h) \neq \emptyset \},
\end{align}
and
\begin{align} \label{est:Lepski_HT}
    \hat{\tau}_{\text{LepskiHT}}  = \frac{1}{n} \sum_{i=1}^n \frac{\mathbf{1}\{ Z_i = 1, e_i \geq \hat{h}_\text{Lepski}\}}{\mathbb{P}\{ Z_i = 1, e_i \geq \hat{h}_\text{Lepski}\}} Y_i - \frac{1}{n} \sum_{i=1}^n \frac{\mathbf{1}\{ Z_i = 0, e_i \leq 1 - \hat{h}_\text{Lepski}\}}{\mathbb{P}\{ Z_i = 0, e_i \leq 1 - \hat{h}_\text{Lepski}\}} Y_i
\end{align}

\subsection{An equivalent formulation of the Horvitz-Thompson estimator for exposure threshold $h$} \label{appendix:bandwidth}

We can rewrite the display in \ref{HT_estimator} in the following form:
\begin{equation} \label{HT_estimator_bandwidth}
    \hat{\tau}_h = \frac{1}{n} \sum_{i=1}^n \frac{\mathbf{1}\{| z_i - e_i | \leq h \}}{\mathbb{P}\{ |z_i - e_i | \leq h\}} \Tilde{z}_i Y_i
\end{equation} where $\Tilde{z}_i = 2z_i -1$, so that $\Tilde{z}_i \in \{-1,1 \}.$ It is then clear that the threshold $h$ controls how much dissimilarity between the treatment status of units and their neighbors, we allow in our estimation. This allows us to reframe the problem as an optimal bandwidth selection one.

\subsection{Bias and Variance estimation errors} \label{appendix:bias_var_error}

In this subsection, we write down the estimation errors of the bias and the variance terms in the MSE estimation.

For the bias terms, we have,
\begin{align*}
    \hat{b}(\hat{\tau}_h) - b^*(\hat{\tau}_h) &= \Big(\frac{1}{n} \sum_{i=1}^n \frac{\hat{\gamma}_n (1-e_i)\mathbf{1}\{Z_i = 1, e_i \geq h \}}{\mathbb{P}(Z_i = 1, e_i \geq h)} - \frac{1}{n} \sum_{i=1}^n \sum_{x_i \in X_i} \frac{\mathbf{1}\{x_i \geq h\}}{|x_i: x_i \in X_i \cap x_i \geq h |}\gamma_n^* (1-x_i) \Big)\\
    &+ \Big( \frac{1}{n} \sum_{i=1}^n \frac{\hat{\gamma}_n e_i\mathbf{1}\{Z_i = 0, e_i \leq 1- h \}}{\mathbb{P}(Z_i = 0, e_i \leq 1-h)} - \frac{1}{n} \sum_{i=1}^n \sum_{x_i \in X_i} \frac{\mathbf{1}\{x_i \leq 1-h\}}{|x_i: x_i \in X_i \cap x_i \leq 1-h |}\gamma_n^* (x_i) \Big),
\end{align*} where $\gamma$ is the slope of the best average linear fit, and where $x_i$ ranges over the set $X_i$ of possible fractions of degree $i$.

By an abuse of notation, we use $y_i(h^+)$ to represent the average (across possible exposure fractions for node $i$) potential outcome for unit $i$ with exposures that are at least $h$, while we use $y_i(h^-)$ to represent the average (across possible exposure fractions for node $i$) potential outcome for unit $i$ with exposures that are at most $1-h$.

The true variance is 
\begin{align*}
    v^*(\hat{\tau}_h) &= \sum_{i=1}^n \frac{y_i(h^+)^2}{n^2} \left(\frac{1}{\pi_i^1} - 1\right) \\
    &+ \sum_{i=1}^n \sum_{\mathclap{\substack{j=1\\ j\neq i}}}^n \frac{y_i(h^+)y_j(h^+)}{n^2} \left(\frac{\pi_{ij}^{11}}{\pi_i^1 \pi_j^1} - 1\right) \\
    &+ \sum_{i=1}^n \frac{y_i(h^-)^2}{n^2\pi_i^0} \left(\frac{1}{\pi_i^0} - 1\right)\\
    &+ \sum_{i=1}^n \sum_{\mathclap{\substack{j=1\\j\neq i}}}^n \frac{y_i(h^-)y_j(h^-)}{n^2} \left(\frac{\pi_{ij}^{00}}{\pi_i^0 \pi_j^0} - 1\right) \\
    &- \frac{2}{n^2} \sum_{i=1}^n \sum_{j\in [n]; \pi_{ij}^{10} > 0} y_i(h^+) y_j(h^-)\left(\frac{\pi_{ij}^{10}}{\pi_i^1 \pi_j^0} - 1\right) \\
    &+ \frac{2}{n^2} \sum_{i=1}^n \sum_{j\in [n]; \pi_{ij}^{10} = 0} y_i(h^+)^2y_j(h^-)^2.
\end{align*}

Therefore, from the above and Section \ref{appendix:variance_estimator}, we have that 

\begin{align*}
    &\sup_h \hat{v}(\hat{\tau}_h) - v^*(\hat{\tau}_h) \\
    &= \sup_h \Bigg[
        \Bigg( \sum_{i=1}^n 
            \frac{\mathbf{1}\{z_i = 1, e_i \geq h\}Y_i^2}{n^2\pi_i^1} 
            \left(\frac{1}{\pi_i^1} - 1\right) 
            - \sum_{i=1}^n \frac{y_i(h^+)^2}{n^2} 
            \left(\frac{1}{\pi_i^1} - 1\right) 
        \Bigg) \\
    &\quad + \Bigg( \sum_{i=1}^n \sum_{\substack{j=1\\ j\neq i}}^n 
            \frac{\mathbf{1}\{z_i = 1, e_i \geq h\}\mathbf{1}\{z_j = 1, e_j \geq h\}Y_iY_j}{n^2\pi_{ij}^{11}} 
            \left(\frac{\pi_{ij}^{11}}{\pi_i^1 \pi_j^1} - 1\right) \\
    &\quad\quad - \sum_{i=1}^n \sum_{\substack{j=1\\ j\neq i}}^n 
            \frac{y_i(h^+)y_j(h^+)}{n^2} 
            \left(\frac{\pi_{ij}^{11}}{\pi_i^1 \pi_j^1} - 1\right) 
        \Bigg) \\
    &\quad + \Bigg( \sum_{i=1}^n 
            \frac{\mathbf{1}\{z_i = 0, e_i \leq 1-h\}Y_i^2}{n^2\pi_i^0} 
            \left(\frac{1}{\pi_i^0} - 1\right) 
            - \sum_{i=1}^n \frac{y_i(h^-)^2}{n^2\pi_i^0} 
            \left(\frac{1}{\pi_i^0} - 1\right) 
        \Bigg) \\
    &\quad + \Bigg( \sum_{i=1}^n \sum_{\substack{j=1\\ j\neq i}}^n 
            \frac{\mathbf{1}\{z_i = 0, e_i \leq 1-h\}\mathbf{1}\{z_j = 0, e_j \leq 1-h\}Y_iY_j}{n^2\pi_{ij}^{00}} 
            \left(\frac{\pi_{ij}^{00}}{\pi_i^0 \pi_j^0} - 1\right) \\
    &\quad\quad - \sum_{i=1}^n \sum_{\substack{j=1\\ j\neq i}}^n 
            \frac{y_i(h^-)y_j(h^-)}{n^2} 
            \left(\frac{\pi_{ij}^{00}}{\pi_i^0 \pi_j^0} - 1\right) 
        \Bigg) \\
    &\quad - \frac{2}{n^2} \sum_{i=1}^n \sum_{\substack{j\in [n]\\ \pi_{ij}^{10} > 0}} 
            \Bigg( 
                \mathbf{1}\{z_i = 1, e_i \geq h\}\mathbf{1}\{z_j = 0, e_j \leq 1-h\}Y_i Y_j 
                \left(\frac{1}{\pi_i^1 \pi_j^0} - \frac{1}{\pi_{ij}^{10}}\right) \\
    &\quad\quad + y_i(h^+) y_j(h^-) 
            \left(\frac{\pi_{ij}^{10}}{\pi_i^1 \pi_j^0} - 1\right) 
        \Bigg) \\
    &\quad + \frac{2}{n^2} \sum_{i=1}^n \sum_{\substack{j\in [n]\\ \pi_{ij}^{10} = 0}} 
            \Bigg( 
                \frac{\mathbf{1}\{z_i = 1, e_i \geq h\}Y_i^2}{2\pi_i^1} 
                + \frac{\mathbf{1}\{z_j = 0, e_j \leq 1-h\}Y_j^2}{2\pi_j^0} \\
    &\quad\quad - y_i(h^+)^2 y_j(h^-)^2 
        \Bigg) 
    \Bigg].
\end{align*}

\subsection{Proofs to Theorem \ref{main_thm} and Corollary \ref{corollary:approx_error}}

\subsubsection{Proof to Theorem \ref{main_thm}}

\begin{proof}[Proof to Theorem \ref{main_thm}] \label{appendix:proof_main}
We first consider the variance terms. Define $\Bar{v} = \mathbb{E}[\hat{v}].$ We have that,
\begin{align*}
    \mathbb{P} (| \hat{v} - v^* | > \Delta_n/4) &=  \mathbb{P} (| \hat{v} - \Bar{v} + \Bar{v} - v^*| > \Delta_n/4)\\
    &\overset{(a)}{=} \mathbb{P} (| \hat{v} - \Bar{v} | > \Delta_n/4 - c d_{max}^2/n ) +   \\
    &\overset{(b)}{\leq} 6\exp{\left( - \frac{C n \left(\frac{\Delta_n}{4} - \frac{c d_{max}^2 }{n}\right)}{\sqrt{A_{n,2}} + \sqrt{\left(\frac{\Delta}{4} - \frac{c d_{max}^2}{n} \right)M_{n,2}}}\right)}
\end{align*} where $A_{n,p} \leq 16 \| \Tilde{v} \|_{L_1}^2\{c_1 + \frac{(\log n)^4}{n}\}^2$, for some constant $c_1$ and $\Tilde{v}$ is the Fourier transform of the variance summands, and $M_{n,2} = 4 \| \Tilde{v} \|_{L_1} (\log n)^4.$ 

The equality $(a)$ is obtained as a result of Condition \ref{cond:bounded_first}. Indeed, we know that $\Bar{v} - v^*$ is non-zero only when there exists $i,j$ such that $\pi_{ij}(h^+, h^-) = 0$, i.e. the joint exposure assignment probability of $e_i \geq h$, and $e_i \leq 1-h$. Since the  contribution of these to $\Bar{v} - v^*$ are bounded above by $c d_{max}^2/n$ for some constant $c$, we get $(a)$. The inequality $(b)$ is obtained from a direct application of \citep[Theorem 2.1]{shen2020exponential}, as it easy to see that our exposure dependencies satisfy $\alpha$-mixing conditions and that together with Condition \ref{cond:positivity}, the variance summands satisfy the Fourier transform conditions of \citep[Theorem 2.1]{shen2020exponential}.

Next, we consider the bias terms, for any constant $J$,
\begin{align*}
    \mathbb{P} ( (\hat{\gamma}_n - \gamma_n^*)^2 > \Delta_n/J) \leq \exp{\left( -\frac{\Delta_n n p(1-p) }{J c d_{\max} + 1 }\right)},
\end{align*} for some constant $c$.
We use the results from \citep[Theorem 3.1]{ziemann2024noise}. Below, we verify that the conditions are met in our setting. We begin with the condition that for every $v \in \mathbb{R}$ such that $v = 1/(\mathbb{E}[e_i^2])^{1/2}$, there exists $h \in \mathbb{R}^+$ such that, $\mathbb{E}[(ve_i)^4] \leq h^2 v^2 \mathbb{E}[e_i^2]$. This is trivially satisfied in our setting since our exposures is bounded. Next, we consider the first part of 
\citep[Condition (3.3)]{ziemann2024noise}. This is also trivially satisfied in our setting since our block sizes ($ = d_{max}^2)$ are bounded by Assumption \ref{cond:bounded_first}. The second part of 
\citep[Condition (3.3)]{ziemann2024noise} is also satisfied since our de-meaned noise-class interaction variables (as defined in \citep[Eq.(2.5)]{ziemann2024noise}) are sub-gaussian by assumption of independence between the potential outcomes' residuals and the treatment design. Next, \citep[Condition (3.4)]{ziemann2024noise}
is also trivially satisfied since we take our blocks to be of size $d_{max}^2$ almost, with the possible exception of the final remaining block, uniformly. Finally, it is easy to see that \citep[Condition (3.5)]{ziemann2024noise} is also satisfied since the exposure variables are independent outside of the radius $d_{max}^2$, since exposure variables $e_i$ and $e_j$ for any $i, j \in [n]$ are only dependent when nodes $i$ and $j$ share a neighbor. We quickly note that the monotone partitioning (as described in \citep[Theorem 3.1]{ziemann2024noise} is not necessary in our setting, as the ``blocking" procedure introduced by \citep{yu1994rates} holds more generally.

For convenience, in the following we write $b^*,\hat{b}$ to mean $b^*(\hat{\tau}_h),\hat{b}(\hat{\tau}_h)$, and similarly $v^*,\hat{v}$ to mean $v^*(\hat{\tau}_h),\hat{v}(\hat{\tau}_h)$.  D is such that $D \geq |H|$, where $H$ the set of exposures. Therefore, putting these together by a union bound,
\begin{align*}
    \mathbb{P}\{\hat{h} = h^*\} 
    &\leq \sum_{h} \mathbb{P}\{| \hat{M}_n(h) - M_n^*(h) | > \Delta_n/ 2\} \\
    &\leq D \mathbb{P}\{| \hat{b}^2 + \hat{v} - {b^*}^2 - v^* | > \Delta_n / 2\} \\
    &\leq D \big(\mathbb{P}\{| \hat{b}^2 - {b^*}^2| > \Delta_n / 4\} + \mathbb{P}\{| \hat{v} - v^* | > \Delta_n / 4\}\big) \\
    &\leq D \big(\mathbb{P}\{(\hat{b} - b^*)^2 > \Delta_n / 8\} + \mathbb{P}\{(\hat{b} - b^*)^2 > \Delta_n^2 /(16 U_n)^2\} \\
    &\quad + \mathbb{P}\{| \hat{v} - v^* | > \Delta_n / 4\} \big) \\
    &\leq D \big(\mathbb{P}\{(\hat{\gamma}_n - \gamma_n^*)^2 > \Delta_n / 8\} + \mathbb{P}\{(\hat{\gamma}_n - \gamma_n^*)^2 > \Delta_n^2 /(16 U_n)^2\} \\
    &\quad + \mathbb{P}\{| \hat{v} - v^* | > \Delta_n / 4\}\big) \\
    &\leq 3D \max \bigg\{
        \exp\bigg( -\frac{\Delta_n n p(1-p)}{8 c d_{\max}} + 1 \bigg), \\
    &\quad \exp\bigg( -\frac{ \Delta_n^2 n p(1-p)}{(16U_n)^2 c d_{\max}} + 1 \bigg), \\
    &\quad 6\exp\bigg( - \frac{C n \big(\frac{\Delta_n}{4} - \frac{c d_{max}^2 }{n}\big)}{
        \sqrt{A_{n,2}} + \sqrt{\big(\frac{\Delta_n}{4} - \frac{c d_{max}^2}{n} \big)M_{n,2}}
    }\bigg)
    \bigg\}.
\end{align*}

The proof for the cluster-level Bernoulli randomization follows immediately.

\end{proof}

\subsubsection{Proof to Corollary \ref{corollary:approx_error}}
\begin{proof}[Proof of Corollary \ref{corollary:approx_error}]
   Suppose $ \sup_{e_i}| f(e_i) - \gamma^*| \leq \delta$. Let $\Tilde{\delta}:=16 \delta^2 + 8 \delta U_n^*$ . We begin by considering the difference between the MSE under the true best average linear fit and the true MSE:
   \begin{align*}
       | M_n^*(h) - M_n^{**}(h)| &= | {b_n^*}^2(h) - {b_n^{**}}^2(h)| \\
       &\leq | (b_n^*(h) - b_n^{**}(h))(b_n^*(h) - b_n^{**}(h))| \\
    &\leq | (b_n^*(h) - b_n^{**}(h))\big((b_n^*(h)- b_n^{**}(h)) + 2b_n^{**}(h) \big)| \\
    &\leq \big((b_n^*(h) - b_n^{**}(h))\big)^2 + 2 |b_n^*(h)- b_n^{**}(h)| U_n^*\\
    &\leq (4\delta)^2 + 2  U_n^* (4 \delta) \equiv \Tilde{\delta}.
   \end{align*}
   This gives us,
    \begin{align*}
        &\mathbb{P}\{ \hat{h}_n \neq  \hat{h}_n^{**} \} \\
        &\leq \sum_{h \in H} \mathbb{P}\{ |\hat{M}_n(h) -  M_n^{**}(h)| > \Delta/2 \}\\
        &= \sum_{h \in H} \mathbb{P}\{ |\hat{M}_n(h) -  M_n^{*}(h) + M_n^{*}(h) - M_n^{**}(h)| > \Delta/2 \}\\
        &\leq \sum_{h \in H} \mathbb{P}\{ |\hat{M}_n(h) -  M_n^{*}(h) | + |M_n^{*}(h) - M_n^{**}(h)| > \Delta/2 \}\\
        &\leq \sum_{h \in H} \mathbb{P}\{ |\hat{M}_n(h) -  M_n^{*}(h) | > \Delta/2 - \Tilde{\delta} \}
\end{align*} where the second inequality is obtained by applying a triangle inequality.

Finally, applying Theorem \ref{main_thm} gives us the result.
\end{proof}


\subsection{Simulations}
\label{appendix:simulations_technical}

\subsubsection{Experimental details on the Amazon product similarity graph data}

The simulations were run on a CPU. Our experiments focus on a subset of the 19828-node DVD graph. In particular, we considered the 17924-node subgraph by removing all isolated nodes. To compute the exposure probabilities, we used $10^6$ simulations. We then selected the first 1000 nodes of the 19828 to analyze. 200 replicates were run, generating random treatment assignments and the corresponding outcomes. For each of the replicates, 1000 separate simulation runs were generated to compute the oracle MSEs. In total, this took approximately 2 hours. The graph data is available at: https://snap.stanford.edu/data/amazon-meta.html

\subsubsection{Experimental details on simulated graphs}
All synthetic graph simulations were run on a machine of Intel® Xeon® processors with 48 CPU cores, and 50GB of RAM. We simulated 1000 replicates, generating random treatment assignments and the corresponding outcomes, with each oracle MSE computed using 1000 separate simulation runs. The exposure probabilities under each threshold were computed as proposed in \citep{aronow2017estimating} using 1000 simulation iterations. In total, this took approximately 30 minutes on average (across the different potential outcomes, and graph settings). 


\subsubsection{Simulations for non-linear potential outcomes models} \label{appendix:non-linear}
We investigate the robustness of our estimator to potential outcomes models that are non-linear in the exposure. In particular, we consider simulations from the sigmoid, and sine exposure functions. In Figure \ref{fig:HT_global_sine}, we display the performance of our estimator under a sigmoid (left) and sine (right) interference function, respectively. Our adaptive threshold Horvitz-Thompson estimator (AdaThresh) improves upon other existing Horvitz-Thompson estimators. 

\begin{figure}
\centering
\includegraphics[width=\linewidth]{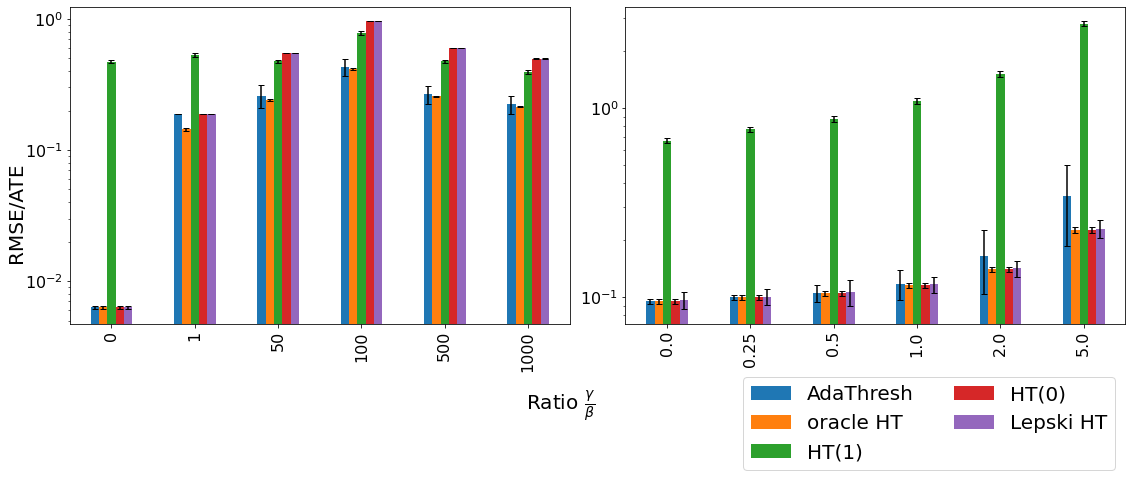}
 \caption{RMSE (normalized by the ATE) induced by the different Horvitz-Thompson estimators. Left: 2nd-power cycle graph under unit-level Ber(0.5) randomization under sigmoid $f(e_i) =\gamma((1 + \exp{(-e_i)}))$. Right: 2nd-power cycle graph under unit-level Ber(0.5) randomization with $f(e_i) = \gamma(1-\sin(\pi\cdot e_i)))$ being a sine function. We focus on varying the ratio $\gamma/\beta$ as we consider a fixed graph, with $n=1000.$ The error bars are two times the standard deviation.}
 \label{fig:HT_global_sine}
\end{figure}

\subsubsection{Simulations for Difference-in-Means estimator with adaptive exposure thresholds}
\label{appendix:difference-in-means}
We consider our approach using the Difference-in-Means estimator incorporating exposure thresholds:
    \begin{align} \label{est:DiM}
    \hat{\tau}_{\text{DiM}_h} = \frac{\sum_{i=1}^n \mathbf{1}\{ Z_i = 1, e_i\geq h\} Y_i}{\sum_{i=1}^n \mathbf{1}\{ Z_i = 1, e_i \geq h\}}  - \frac{\sum_{i=1}^n \mathbf{1}\{ Z_i = 0, e_i \leq 1-h\} Y_i}{\sum_{i=1}^n \mathbf{1}\{ Z_i = 0, e_i \leq 1-h\}}. 
\end{align}

We use the following bias estimator:
\begin{align*}
     \hat{b}(\hat{\tau}_h) = \sum_{i=1}^n \frac{(1-e_i) \hat{\gamma}_n\mathbf{1}\{ Z_i = 1, e_i \geq h\}}{\sum_{i=1}^n \mathbf{1}\{ Z_i = 1, e_i \geq h\}} + \sum_{i=1}^n \frac{e_i \hat{\gamma}_n\mathbf{1}\{ Z_i = 0, e_i \leq  1- h\}}{\sum_{i=1}^n \mathbf{1}\{ Z_i = 0, e_i \leq 1-h\}},
\end{align*} where $\hat{\gamma}$ is the linear regression coefficient for the exposure variable.

We use the following variance estimator, decomposing it into its treatment and control parts:
\begin{align*}
    \hat{v}(\hat{\tau}) = \frac{2}{n-1}\left( \hat{v}_T + \hat{v}_C \right)
\end{align*} where 
\begin{align*}
    \hat{v}_T = \frac{1}{n_1} \sum_{i=1}^n \left( Y_i \mathbf{1}\{ Z_i = 1, e_i \geq h\} - \frac{1}{n_1} \sum_{i=1}^n Y_i \mathbf{1}\{ Z_i = 1, e_i \geq h\} \right)^2
\end{align*} where $n_1 := \sum_{i=1}^n \mathbf{1}\{ Z_i = 1, e_i \geq h\} $, and 
\begin{align*}
    \hat{v}_C = \frac{1}{n_0} \sum_{i=1}^n \left( Y_i \mathbf{1}\{ Z_i = 0, e_i \leq  1-h\} - \frac{1}{n_0} \sum_{i=1}^n Y_i \mathbf{1}\{ Z_i = 0, e_i \leq 1-h\} \right)^2
\end{align*} where $n_0 := \sum_{i=1}^n \mathbf{1}\{ Z_i = 0, e_i \leq 1-h\} $.

We compare the performance of our adaptive estimator to the vanilla difference-in-means estimator, the difference-in-means analogue of the vanilla Horvitz-Thompson estimator, and the difference-in-means estimator with a threshold plugin via Lepski's method. We write these out below:

\begin{itemize}
    \item vanilla difference-in-means estimator
    \begin{align} \label{est:vanilla_DiM}
    \hat{\tau}_{\text{DiM}_0} = \frac{\sum_{i=1}^n \mathbf{1}\{ Z_i = 1\} Y_i}{\sum_{i=1}^n \mathbf{1}\{ Z_i = 1\}}  - \frac{\sum_{i=1}^n \mathbf{1}\{ Z_i = 0\} Y_i}{\sum_{i=1}^n \mathbf{1}\{ Z_i = 0\}} 
\end{align}
    \item difference-in-means estimator at threshold $h=1$
    \begin{align} \label{est:DiM_thresh1}
    \hat{\tau}_{\text{DiM}_1} = \frac{\sum_{i=1}^n \mathbf{1}\{ Z_i = 1, e_i = 1\} Y_i}{\sum_{i=1}^n \mathbf{1}\{ Z_i = 1, e_i = 1\}}  - \frac{\sum_{i=1}^n \mathbf{1}\{ Z_i = 0, e_i = 0\} Y_i}{\sum_{i=1}^n \mathbf{1}\{ Z_i = 0, e_i = 0\}} 
    \end{align}
    \item difference-in-means estimator at threshold $\hat{h}_{\text{Lepski}}$ where
    \begin{align}
    \label{lepski_h_dim}
        \hat{h}_\text{Lepski} := \min \{h \in H : \cap_{h=1}^k I(h) \neq \emptyset \},
    \end{align} with
    \begin{align}
    \label{lepski_interval_dim}
        I(h) := [ \hat{\tau}_{\text{DIM}_h} - 2 \hat{\text{SDEV}}(\hat{\tau}_{\text{DIM}_h}), \hat{\tau}_{\text{DIM}_h} + 2 \hat{\text{SDEV}}(\hat{\tau}_{\text{DIM}_h})],
    \end{align} and
    \begin{align} \label{est:Lepski_dim}
        \hat{\tau}_{\text{LepskiDiM}}  = \sum_{i=1}^n \frac{\mathbf{1}\{ Z_i = 1, e_i \geq \hat{h}_\text{Lepski}\}}{\sum_{i=1}^n \mathbf{1}\{ Z_i = 1, e_i \geq\hat{h}_\text{Lepski}\}} Y_i - \sum_{i=1}^n \frac{\mathbf{1}\{ Z_i = 0, e_i \leq 1 - \hat{h}_\text{Lepski}\}}{\sum_{i=1}^n \mathbf{1}\{ Z_i = 0, e_i \leq 1 - \hat{h}_\text{Lepski}\}} Y_i
    \end{align}
\end{itemize}
In Figures \ref{fig:GLR_DIM} and \ref{fig:DIM_global_sine}, we display the performance of our adaptive threshold (AdaThresh) Difference-in-Means estimators in comparison to other estimators. AdaThresh improves upon fixed threshold Difference-in-Means estimators.

\begin{figure}
\includegraphics[width=\textwidth]{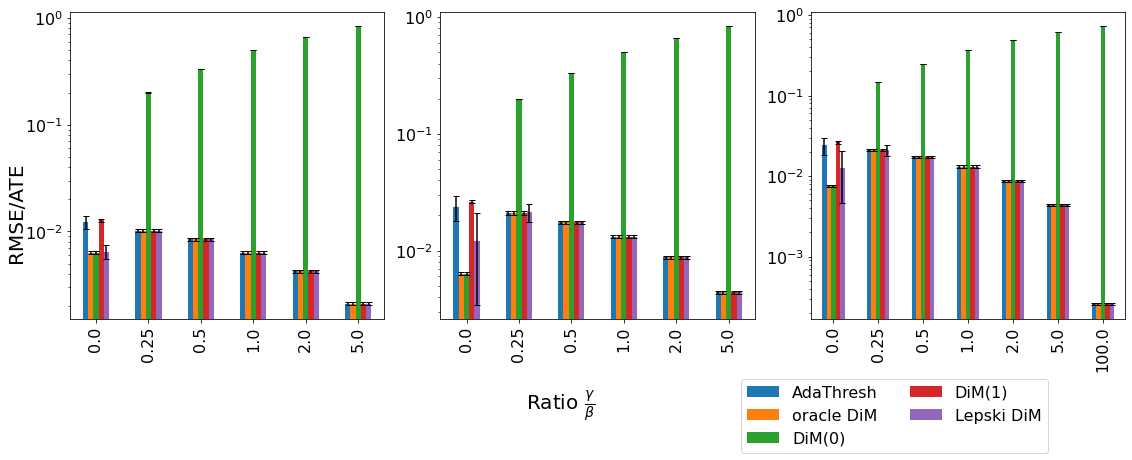}
\caption{RMSE (normalized by the ATE) under the linear model with $\alpha =10$, $g(z_i) = \beta z_i = 10z_i$, $f(e_i) =\gamma e_i$, and fixed $\epsilon_i$ generated from $\mathcal{N}(0,1),$ induced by the different Difference-in-Means estimators. We focus on varying the ratio $\gamma/\beta$ as we consider a fixed graph. Left: Cycle graph under unit-level Ber(0.5) randomization. Center: 2nd-power cycle graph under unit-level Ber(0.5) randomization. Right: 2nd-power cycle graph under cluster-level Ber(0.5) randomization with cluster sizes 5 (=2k + 1). The error bars are two times the standard deviation.}
\label{fig:GLR_DIM}
\end{figure}

\begin{figure}
\centering
\includegraphics[width=\textwidth]{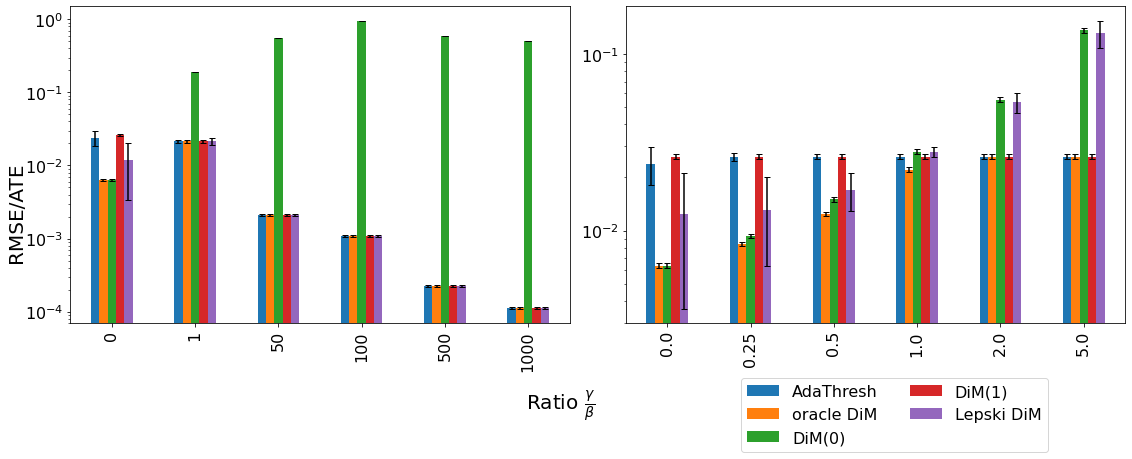}
\caption{RMSE (normalized by the ATE) induced by the different Difference-in-Means estimators. Left: 2nd-power cycle graph under unit-level Ber(0.5) randomization under  sigmoid $f(e_i) =\gamma(1 + \exp{(-e_i)}))$.  Right: 2nd-power cycle graph under unit-level Ber(0.5) randomization with $f(e_i) = \gamma(1-\sin(\pi\cdot e_i)))$ being a sine function. We focus on varying the ratio $\gamma/\beta$ as we consider a fixed graph, with $n=1000.$ The error bars are two times the standard deviation.}
\label{fig:DIM_global_sine}
\end{figure}

\subsubsection{Simulations using local linear regression}
\label{appendix:localLR}

We extend our global linear regression approach to a local linear regression one to estimate the rate of change of the bias. We write the new bias estimators for the Horvitz-Thompson and Difference-in-Means estimators. We illustrate the performance of the local linear regression in the settings above, as well as settings that significantly deviates from linearity.

We write the bias estimator for the Horvitz-Thompson estimator as:
 \begin{align*}
     \hat{b}(\hat{\tau}_h) = \frac{1}{n} \sum_{i=1}^n \frac{(1-e_i) \hat{\gamma}_n^{(h)}\mathbf{1}\{ Z_i = 1, e_i \geq h\}}{\mathbb{P}\{ Z_i = 1, e_i \geq h\}} + \frac{1}{n} \sum_{i=1}^n \frac{e_i \hat{\gamma}_n^{(1-h)}\mathbf{1}\{ Z_i = 0, e_i \leq  1- h\}}{\mathbb{P}\{ Z_i = 0, e_i \leq 1- h\}},
 \end{align*} where $\hat{\gamma}_n^{(h)}$ and $\hat{\gamma}_n^{(1-h)}$ are the local linear regression coefficients for the exposure variable in the intervals $[h,1]$, and $[0,1-h]$, respectively.

Similarly, write the bias estimator for the Difference-in-Means estimator as:
\begin{align*}
     \hat{b}(\hat{\tau}_h) = \sum_{i=1}^n \frac{(1-e_i) \hat{\gamma}_n^{(h)}\mathbf{1}\{ Z_i = 1, e_i \geq h\}}{\sum_{i=1}^n \mathbf{1}\{ Z_i = 1, e_i \geq h\}} + \sum_{i=1}^n \frac{e_i \hat{\gamma}_n^{(1-h)}\mathbf{1}\{ Z_i = 0, e_i \leq  1- h\}}{\sum_{i=1}^n \mathbf{1}\{ Z_i = 0, e_i \leq 1-h\}},
\end{align*} where $\hat{\gamma}_n^{(h)}$ and $\hat{\gamma}_n^{(1-h)}$ are the local linear regression coefficients for the exposure variable in the intervals $[h,1]$, and $[0,1-h]$, respectively.

\begin{figure}
\centering
\includegraphics[width=\linewidth]{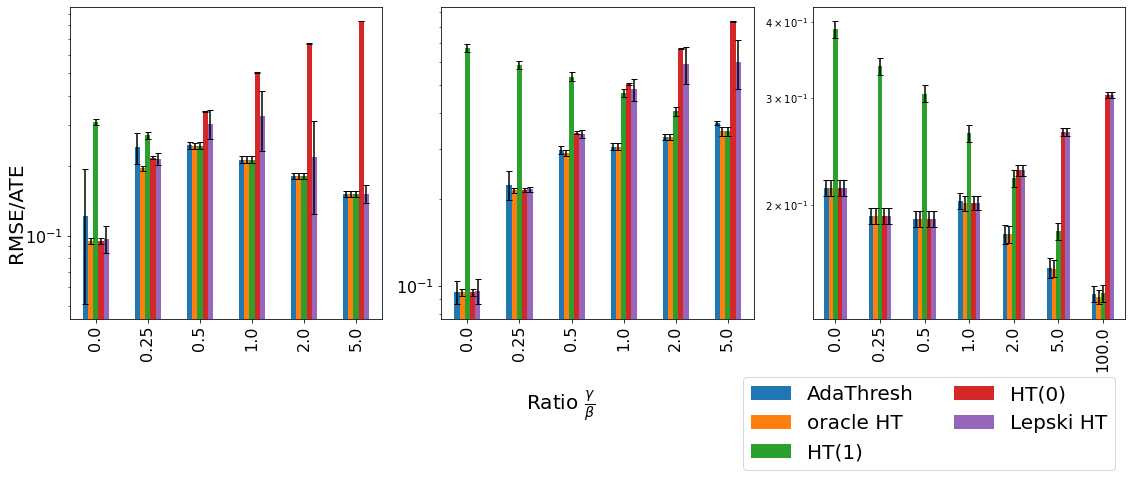}
\caption{RMSE (normalized by the ATE) under the linear model with $\alpha =10$, $g(z_i) = \beta z_i = 10z_i$, $f(e_i) =\gamma e_i$, and fixed $\epsilon_i$ generated from $\mathcal{N}(0,1),$ induced by the different (local) Horvitz-Thompson estimators. We focus on varying the ratio $\gamma/\beta$ as we consider a fixed graph, with $n=1000.$ Left: Cycle graph under unit-level Ber(0.5) randomization. Center: 2nd-power cycle graph under unit-level Ber(0.5) randomization. Right: 2nd-power cycle graph under cluster-level Ber(0.5) randomization with cluster sizes 5 (=2k + 1). The error bars are two times the standard deviation.}
\label{fig:HT_local}
\end{figure}

\begin{figure}
\centering
\includegraphics[width=\linewidth]{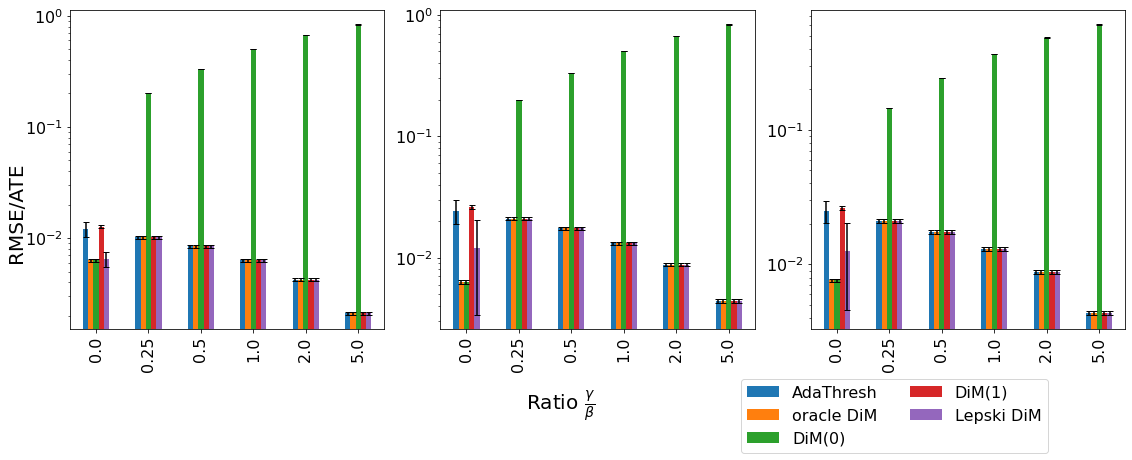}
\caption{RMSE (normalized by the ATE) under the linear model with $\alpha =10$, $g(z_i) = \beta z_i = 10z_i$, $f(e_i) =\gamma e_i$, and fixed $\epsilon_i$ generated from $\mathcal{N}(0,1),$ induced by the different (local) Difference-in-Means estimators. We focus on varying the ratio $\gamma/\beta$ as we consider a fixed graph, with $n=1000.$ Left: Cycle graph under unit-level Ber(0.5) randomization. Center: 2nd-power cycle graph under unit-level Ber(0.5) randomization. Right: 2nd-power cycle graph under cluster-level Ber(0.5) randomization with cluster sizes 5 (=2k + 1). The error bars are two times the standard deviation.}
\label{fig:DIM_local}
\end{figure}

Figures \ref{fig:HT_local}, \ref{fig:HT_local_sine}, \ref{fig:DIM_local}, and \ref{fig:DIM_local_sine}, display how local linear regression bias estimates perform.

\begin{figure}
\centering
\includegraphics[width=\linewidth]{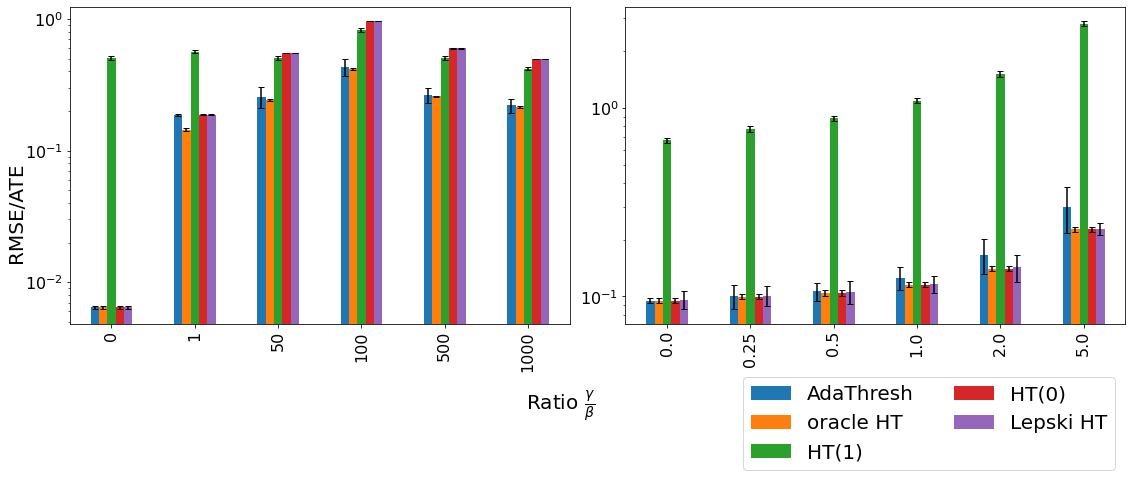}
\caption{RMSE (normalized by the ATE) under the linear model with $\alpha =10$, $g(z_i) = \beta z_i = 10z_i$, and fixed $\epsilon_i$ generated from $\mathcal{N}(0,1)$, induced by the different (local) Horvitz-Thompson estimators. Left: 2nd-power cycle graph under unit-level Ber(0.5) randomization under sigmoid $f(e_i) =\gamma(1 + \exp{(-e_i)}))$. Right: 2nd-power cycle graph under unit-level Ber(0.5) randomization with $f(e_i) = \gamma(1-\sin(\pi\cdot e_i)))$ being a sine function. We focus on varying the ratio $\gamma/\beta$ as we consider a fixed graph, with $n=1000.$ The error bars are two times the standard deviation.}
\label{fig:HT_local_sine}
 \end{figure}

\begin{figure}
\centering
\includegraphics[width=\textwidth]{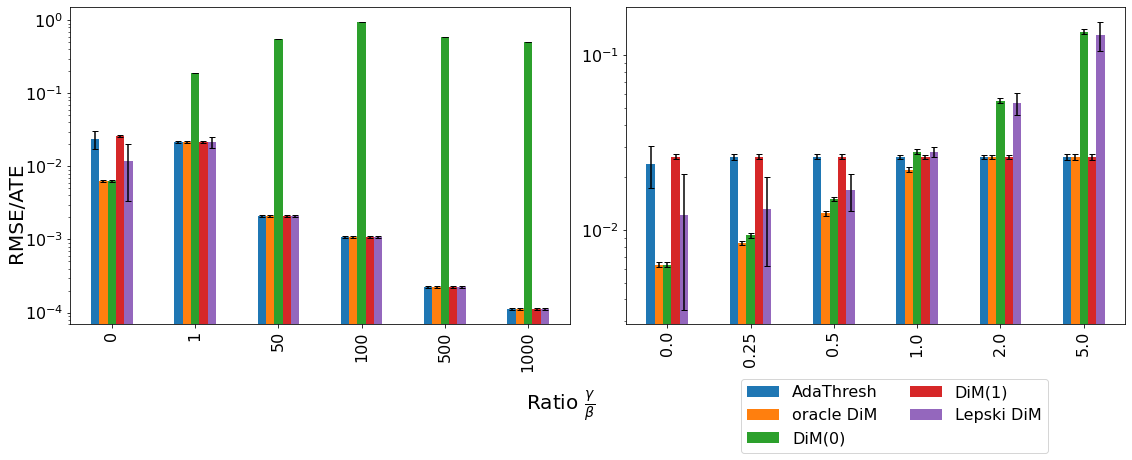}
\caption{RMSE (normalized by the ATE) induced by the different (local) Difference-in-Means estimators. Left: 2nd-power cycle graph under unit-level Ber(0.5) randomization under  sigmoid $f(e_i) =\gamma(1 + \exp{(-e_i)}))$. Right: 2nd-power cycle graph under unit-level Ber(0.5) randomization with $f(e_i) = \gamma(1-\sin(\pi\cdot e_i))$ being a sine function. We focus on varying the ratio $\gamma/\beta$ as we consider a fixed graph, with $n=1000.$ The error bars are two times the standard deviation.}
\label{fig:DIM_local_sine}
\end{figure}

\subsection{More on Tradeoffs in Circulant graphs: Cycles, and $k$th power Cycle} \label{circulant_more}

We further consider cluster-level Bernoulli($p$) randomization, with clusters of size $2k + 1$, in the cycle graph and the $k$th-power cycle graphs. We say a graph is a $k$th power-Cycle graph if there exists an edge between each node and $2k$ of its nearest neighbors \citep{ugander2013graph}. 



\begin{figure}
     \centering
     \begin{subfigure}[b]{0.3\linewidth}
         \centering
         \includegraphics[height=0.8\textwidth, width=\textwidth]{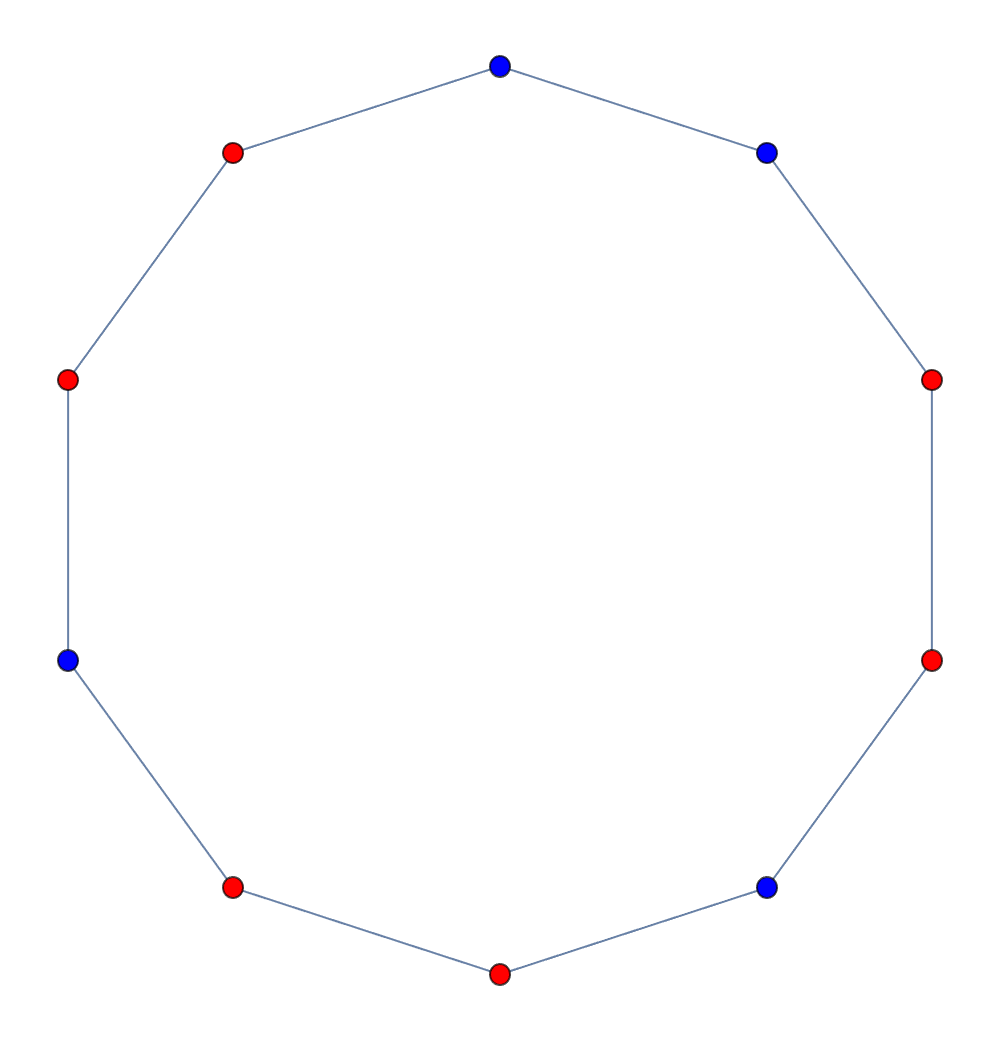}
         \label{fig:cycle_graph}
     \end{subfigure}
    \hfill
     \begin{subfigure}[b]{0.3\textwidth}
         \centering
         \includegraphics[height=0.8\textwidth, width=\textwidth]{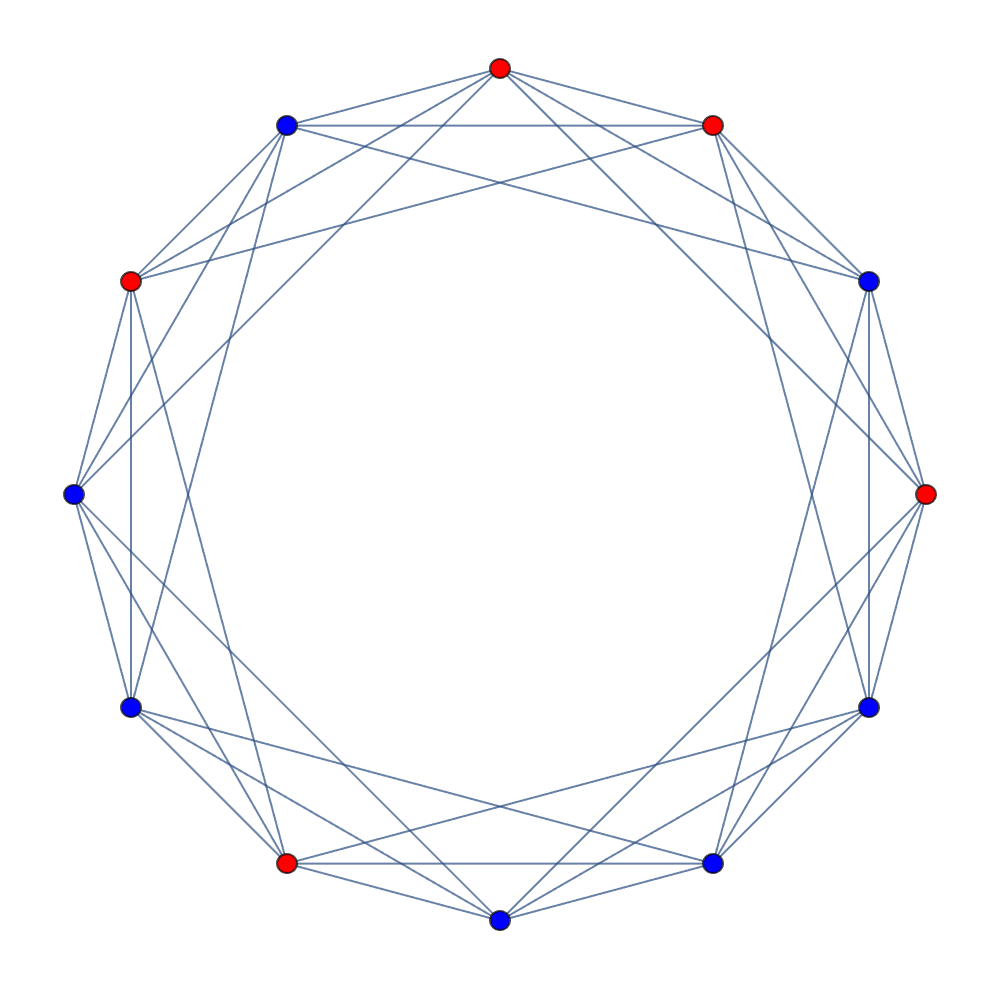}
         \label{fig:3rd_power_cycle}
     \end{subfigure}
     \hfill
     \begin{subfigure}[b]{0.3\textwidth}
         \centering
         \includegraphics[height=0.8\textwidth, width=\textwidth]{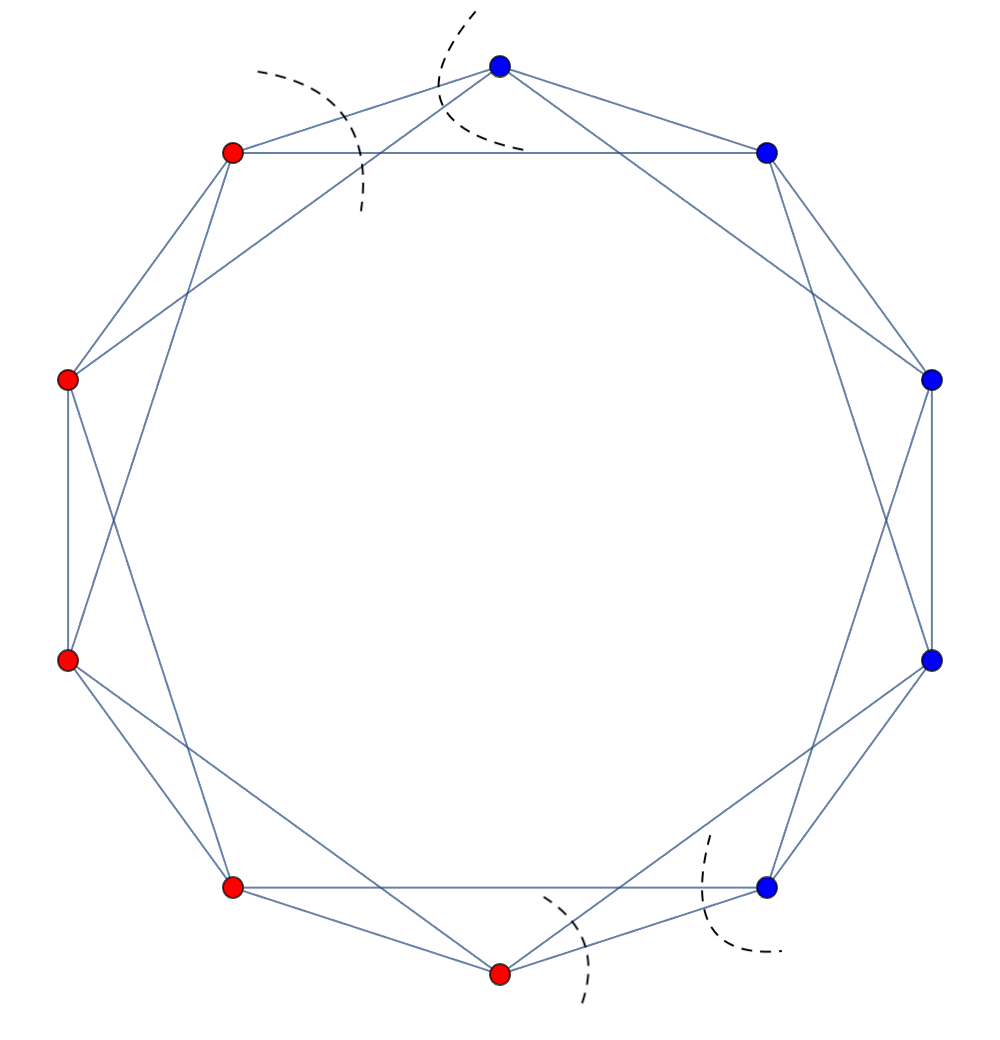}
         \label{fig:2nd_power_cycle}
     \end{subfigure}
     \caption{Circulant graphs with unit-level randomization and cluster-level randomization. Blue and red nodes represent treated and control units, respectively. Left: cycle graph with Ber(0.5) unit  randomization. Center: 3rd-power cycle graph with Benoulli(0.5) unit randomization. Right: 2nd-power cycle graph with Bernoulli(0.5) cluster randomization, with clusters of size 5 (=2k + 1).}
     \label{fig:Circulant_graphs2}
\end{figure}

For $k$th-power cycle graphs, we also consider cluster-randomized design with cluster size $2k + 1$. In \citep{ugander2013graph}, the authors show that this clustering size minimizes variance under full neighborhood exposure. We state the approximate absolute bias and variance for this setting in the following propositions. Here, define the threshold $h = \frac{l}{d}$ with $l=0,1,...,d$.
\begin{proposition}[Absolute bias in k-th power cycle graph under cluster-randomization]
\label{bias_cluster}
When $p=1/2$, and the potential outcome model is simply linear, i.e. $Y_i = \alpha + \beta z_i + \gamma e_i$, the bias of the Horvitz-Thompson estimator for a given threshold $h$ in the k-th power cycle graph under cluster randomization, with cluster-sizes $2k+1$, is approximately
$2\gamma (h-1)$ for $h \geq 1/2$.
\end{proposition}

\begin{proposition}[Variance in the k-th power cycle graph under cluster-randomization]
When $p = 1/2$ and the potential outcome model is simply linear, i.e. $Y_i = \alpha + \beta z_i + \gamma e_i$, the variance of the Horvitz-Thompson estimator for a given threshold $h$, in the $k$-th power cycle graph under cluster-level randomization, with cluster-sizes $2k + 1$, is proportional to 
\label{var_cluster}
\begin{align*}
    \frac{1}{np^2}  (3d + 1 - 2dh) [(\beta + \gamma h)^2 + (\gamma(1-h))^2] = \Theta(\frac{\beta^2h^2}{np^{2h}}).
\end{align*}
\end{proposition}

Under cluster randomization with cluster sizes $2k + 1$ for the k-th power cycle graphs, the variance grows linearly in the degrees of the graph. Therefore, informally, compared to the unit-randomized design setting, higher node degrees lead to stronger bias than variance for a fixed exposure function $f(e_i).$

\subsection{Proofs to Propositions \ref{bias_unit}, \ref{var_unit}, \ref{bias_cluster}, \ref{var_cluster}}

\begin{proof}[Proof to Proposition \ref{bias_unit}]
When $p = 1/2$ and the potential outcome model is simply linear, i.e. $Y_i = \alpha + \beta z_i + \gamma e_i$, the absolute bias of the Horvitz-Thompson estimator for a given threshold $h \equiv l/2k$ for $l=0,2,...,2k$, in the $k$th-power cycle graph under unit-randomization is equal to 
\begin{align*}
    &\frac{1}{\mathbb{P}\{z_ i =1, e_i \geq h \}} \sum_{x_i \in X} \frac{\mathbf{1}\{ x_i \geq h\}}{|x_i: x_i \in X_i \cap x_i \geq h |}\mathbb{E}[\mathbf{1}\{z_i = 1, e_i =x_i\}]y_i(x_i) \\
    &-  \frac{1}{\mathbb{P}\{z_ i =0, e_i \leq 1-h \}} \sum_{x_i \in X} \frac{\mathbf{1}\{ x_i \leq 1-h\}}{|\{x_i: x_i \in X_i \cap x_i \leq 1- h \}|}\mathbb{E}[\mathbf{1}\{z_i = 0, e_i = x_i\}]y_i(x_i) \\
    & - \left(\beta + \gamma \right) \\
    &= \frac{1}{\sum_{r=l}^{2k} \binom{2k}{r}p^{2r}} \sum_{x_i \in X} \frac{\mathbf{1}\{ x_i \geq h\}}{|\{x_i: x_i \in X_i \cap x_i \leq 1- h \}|}\mathbb{E}[\mathbf{1}\{z_i = 1, e_i =x_i\}]y_i(x_i) \\
    &-  \frac{1}{\sum_{r=l}^{2k} \binom{2k}{r}p^{2r}} \sum_{x_i \in X} \frac{\mathbf{1}\{ x_i \leq 1-h\}}{|\{x_i: x_i \in X_i \cap x_i \leq 1- h \}|}\mathbb{E}[\mathbf{1}\{z_i = 0, e_i = x_i\}]y_i(x_i) \\
    &- (\beta + \gamma ) \\
    &= \frac{1}{\sum_{r=l}^{2k} \binom{2k}{r}} \left(\sum_{r=l}^{2k} \binom{2k}{r} \beta  +  \binom{2k}{r} (r/k -1 ) \gamma\right) - (\beta + \gamma ) \\
    &= \frac{1}{\sum_{r=l}^{2k} \binom{2k}{r}} \left(\sum_{r=l}^{2k} \binom{2k}{r} (r/k -1 ) \gamma  \right) -  \gamma \\
\end{align*}
where $X = \{0,1/2k, 2/2k, ..., 1\}.$
\end{proof}

\begin{proof}[Proof to Proposition \ref{var_unit}]
    When $p = 1/2$ and the potential outcome model is simply linear, i.e. $Y_i = \alpha + \beta z_i +\gamma e_i$, it is not difficult to see that the variance of the Horvitz-Thompson estimator for a given threshold $h \equiv l/2k$, in the $k$-th power cycle graph under unit-level randomization is proportional to 
    \begin{align*}
    &\frac{1}{n}  \sum_{l=0}^{2k}  \binom{2k}{l}\mathbf{1}\{ l/2k \geq h \} \left[(\alpha + \beta + \gamma l/2k)^2 (\frac{1}{ \sum_{l=0}^{2k}  \binom{2k}{l}\mathbf{1}\{ l/2k \geq h \}p^{2l}} - 1) \right. \\
    &+ \left. 2 (\beta + \gamma l/(2k))(\gamma (1-l/2k)) (2k+1- \frac{2k}{ \sum_{l=0}^{2k}  \binom{2k}{l}\mathbf{1}\{ l/2k \geq h \}p^{2l}}) \right.\\
    &+ \left. (\alpha + \gamma (1-l/2k))^2(\frac{1}{ \sum_{l=0}^{2k} \binom{2k}{l}\mathbf{1}\{ l/2k \geq h \}p^{2l}} - 1) \right].
    \end{align*}
    Considering the dominating terms gives us the result.
\end{proof}

\begin{proof}[Proof to Proposition \ref{bias_cluster}]
When $p = 1/2$ and the potential outcome model is simply linear, i.e. $Y_i = \alpha + \beta z_i + \gamma e_i$, the absolute bias of the Horvitz-Thompson estimator for a given threshold $h \equiv l/2k$ for $l=k,k+1, ..., 2k$, in the $k$th-power cycle graph under cluster-randomization is equal to is equal to 
\begin{align*}
    &\frac{1}{\mathbb{P}\{z_ i =1, e_i \geq h \}} \sum_{x_i \in X} \frac{\mathbf{1}\{ x_i \geq h\}}{|x_i: x_i \in X_i \cap x_i \geq h |}\mathbb{E}[\mathbf{1}\{z_i = 1, e_i =x_i\}]y_i(x_i) \\
    &-  \frac{1}{\mathbb{P}\{z_ i =0, e_i \leq 1-h \}} \sum_{x_i \in X} \frac{\mathbf{1}\{ x_i \leq 1-h\}}{|x_i: x_i \in X_i \cap x_i \leq 1- h |}\mathbb{E}[\mathbf{1}\{z_i = 0, e_i = x_i\}]y_i(x_i) \\
    & - \left(\beta + \gamma  \right) \\
    &= \frac{(p/(2k+1)+ 2k/(2k+1) \cdot p^2) (\beta + \gamma ) + 2p/(2k + 1)\sum_{r=1}^{k} \mathbf{1}\{ r \geq d-l\}(\beta + \gamma \cdot (2l -d )/d)}{p/(2k+1) + p^2 \cdot 2k/(2k + 1) + 2p/(2k + 1)\sum_{r=1}^{k} \mathbf{1}\{ r \geq d-l\}}\\
    &- (\beta + \gamma ) \\
    &= \mathcal{O}(2\gamma (h-2))
\end{align*}
where $X = \{0,1/2k, 2/2k, ..., 1\}.$ When $l < k$, the bias scales is the same as when $l = k.$
\end{proof}

\begin{proof}[Proof to Proposition \ref{var_cluster}]
    When $p = 1/2$ and the potential outcome model is simply linear, i.e. $Y_i = \alpha + \beta z_i + \gamma e_i$, it is not difficult to see that the variance of the Horvitz-Thompson estimator for a given threshold $h \equiv l/2k$, in the $k$-th power cycle graph under cluster-randomization, with cluster-size $2k +1$, is proportional to 
     \begin{align*}
        &\sum_{u=0}^d \mathbf{1}\{ u/d \geq h\} \binom{d}{u}  \sum_{i=1}^n (\frac{(\beta +\gamma u /d)^2}{n^2}) [(3d + 1 - 2l)(\frac{1}{p^2}-1) + (2l+1)(\frac{1}{p} - 1)] \\
        &+ \sum_{u=0}^d \mathbf{1}\{ u/d \geq h\} \binom{d}{u} \sum_{i=1}^n (\frac{(\gamma (\frac{d-u}{d}))^2}{n^2}) [(3d + 1 - 2l)(\frac{1}{p^2}-1) +  (2l+1)(\frac{1}{p} - 1)] \\
        &- \sum_{u=0}^d \mathbf{1}\{ u/d \geq h\} \binom{d}{u} \frac{2}{n}(\beta + \gamma\cdot u/d)(\gamma\cdot \frac{d-u}{d}). 
    \end{align*}
    Considering the dominating terms gives us the result.
\end{proof}

\subsection{Discussion on ``double-dipping"} \label{double_dipping}

Since our approach involves "double-dipping" into the data, we need to make sure that there is no overfitting that occurs. The authors of \citep{chernozhukov2018double}, \citep{belloni2015program}, and \citep{kennedy2022semiparametric} describe this in more detail. While sample-splitting would simply take care of this in the case of independent data, it does not apply to our setting where the data are dependent, as modeled by the network. One could potentially leverage results under the dependent-data setting, such as \cite{hart1990data}, but this is outside the scope of our paper. We propose to use our approach on the whole sample data and argue this via empirical process theory arguments. In particular, we can think of our threshold $h$ as a nuisance parameter, and we show that our estimator for $h$ has a simple enough associated MSE function class. We note that it is sufficient to show that we are not overfitting by showing that the rate of convergence of the estimated MSE under $\hat{h}$ converges to the true MSE under the optimal threshold, under the best possible linear fit, $h^*$ at least at a $\mathcal{O}(n^{-1/2})$-rate.

\begin{proposition}\label{prop:stochastic_equi}
Suppose that Conditions \ref{cond:positivity}, \ref{cond:bounded_var}, and \ref{cond:bounded_first} are satisfied. The corresponding bias terms in the estimated and true MSE, $\hat{M}_n^b$, and $M_n^b$, under ``double prediction", satisfy
\begin{align*}
     &\mathbb{E}\left[\sup_{\delta/2 < |h_n^*-\hat{h}_n| < \delta} \sqrt{n} \left(\hat{M}_n^b(\hat{h}_n) - M_n^b(\hat{h}_n) - \hat{M}_n^b(h^*) - M_n^b(h^*)\right)\right] \to 0,
\end{align*} as $\delta \to 0$.
\end{proposition} 



Suppose that Conditions \ref{cond:positivity}, \ref{cond:bounded_var}, and \ref{cond:bounded_first} are satisfied. We aim to show that the corresponding biases in the MSE terms under `` double prediction" satisfy
\begin{align*}
     &\mathbb{E}\left[\sup_{\delta/2 < |h^*-h_n| < \delta} \sqrt{n} \left(\hat{M}_n^b(h_n) - M_n^b(h_n) - \hat{M}_n^b(h^*) - M_n^b(h^*)\right)\right] \to 0,
\end{align*} as $\delta \to 0$, where $x_i$ ranges over the set $X_i$ of possible fractions of degree $i$.

\begin{proof}[Proof to Proposition \ref{prop:stochastic_equi}]
\label{appendix:uniform_clt}
We begin by considering the empirical process in question. We write out
\begin{align*}
    \hat{M}_n(h) &= \left(\frac{1}{n} \sum_{i=1}^n \frac{\hat{\gamma}_h(1-e_i)\mathbf{1}\{ Z_i = 1, e_i \geq h \}}{\mathbb{P}\{ Z_i =1, e_i \geq h \}} +  \frac{1}{n} \sum_{i=1}^n \frac{\hat{\gamma}_h(1-e_i)\mathbf{1}\{ Z_i = 1, e_i \geq h \}}{\mathbb{P}\{ Z_i =1, e_i \geq h \}}\right)^2 \\
    &+ \sum_{i=1}^n \frac{\mathbf{1}\{z_i = 1, e_i \geq h\}Y_i^2}{n^2\pi_i^1} (\frac{1}{\pi_i^1} - 1) \\
    &+ \left. \sum_{i=1}^n \sum_{\mathclap{\substack{j=1\\j\neq i}}}^n \frac{\mathbf{1}\{z_i = 1, e_i \geq h\}\mathbf{1}\{z_j = 1, e_j \geq h\}Y_iY_j}{n^2\pi_{ij}^{11}} (\frac{\pi_{ij}^{11}}{\pi_i^1 \pi_j^1} - 1) \right.\\
    &+ \left. \sum_{i=1}^n \frac{\mathbf{1}\{z_i = 0, e_i \leq 1-h\}Y_i^2}{n^2\pi_i^0} (\frac{1}{\pi_i^0} - 1) \right.\\
    &+ \left. \sum_{i=1}^n \sum_{\mathclap{\substack{j=1\\j\neq i}}}^n \frac{\mathbf{1}\{z_i = 0, e_i \leq 1-h\}\mathbf{1}\{z_j = 0, e_j \leq 1-h\}Y_iY_j}{n^2\pi_{ij}^{00}} (\frac{\pi_{ij}^{00}}{\pi_i^0 \pi_j^0} - 1) \right.\\
    &-  \frac{2}{n^2} \sum_{i=1}^n \sum_{j\in [n]; \pi_{ij}^{10} > 0} ( \mathbf{1}\{z_i = 1, e_i \geq h\}\mathbf{1}\{z_j = 0, e_j \leq 1-h\}Y_i Y_j)(\frac{1}{\pi_i^1 \pi_j^0} - \frac{1}{\pi_{ij}^{10}} )  \\
    &+ \frac{2}{n^2} \sum_{i=1}^n \sum_{j\in [n]; \pi_{ij}^{10} = 0} ( \frac{\mathbf{1}\{z_i = 1, e_i \geq h\}Y_i^2}{2 \pi_i^1} + \frac{\mathbf{1}\{ z_j = 0, e_j \leq 1-h \}Y_j^2}{2 \pi_j^0}) .
\end{align*}

The absolute Horvitz-Thompson bias estimate at threshold $h$ (and at $1-h$) is 
\begin{align}
   \hat{b}(\hat{\tau}_h) = \frac{1}{n} \sum_{i=1}^n \frac{\hat{\gamma}_n(1-e_i)\mathbf{1}\{Z_i = , e_i \geq h \}}{\mathbb{P}(Z_i = 1, e_i \geq h)} + \frac{1}{n} \sum_{i=1}^n \frac{\hat{\gamma}_n(e_i)\mathbf{1}\{Z_i = 0, e_i \leq 1 -h \}}{\mathbb{P}(Z_i = 0, e_i \leq 1- h)} 
\end{align}


Then, for every $h$, using the identity $x^2 - y^2 = (x +y)(x-y)$, and defining $U_n \in \mathbb{R}$ such that $U_n \geq M_n^b(h)$ for all $h$ (note that there exists such a $U_n < \infty$ by Conditions \ref{cond:positivity}, \ref{cond:bounded_var}), we have
\begin{align*}
    \hat{M}_n^b(\hat{\tau}_h) - M_n^b(\hat{\tau}_h) \leq 2U_n\left( \hat{b}(\hat{\tau}_h) - b^*(\hat{\tau}_h)\right) + \left( \hat{b}(\hat{\tau}_h) - b^*(\hat{\tau}_h)\right)^2,
\end{align*}

where, from Section \ref{appendix:bias_var_error}, the difference between the bias estimate and the true bias induced by the best average linear fit is
\begin{align*}
    \hat{b}(\hat{\tau}_h) - b^*(\hat{\tau}_h) &= \left(\frac{1}{n} \sum_{i=1}^n \frac{\hat{\gamma} _n(1-e_i)\mathbf{1}\{Z_i = 1, e_i \geq h \}}{\mathbb{P}(Z_i = 1, e_i \geq h)} - \frac{1}{n} \sum_{i=1}^n \sum_{x_i \in X_i} \frac{\mathbf{1}\{x_i \geq h\}}{|x_i: x_i \in X_i \cap x_i \geq h |}\gamma_n (1-x_i) \right)\\
    &+ \left( \frac{1}{n} \sum_{i=1}^n \frac{\hat{\gamma}_n e_i\mathbf{1}\{Z_i = 0, e_i \leq 1- h \}}{\mathbb{P}(Z_i = 0, e_i \leq 1-h)} - \frac{1}{n} \sum_{i=1}^n \sum_{x_i \in X_i} \frac{\mathbf{1}\{x_i \leq 1-h\}}{|x_i: x_i \in X_i \cap x_i \leq 1-h |}\gamma_n(x_i) \right),
\end{align*} where $\gamma_n$ is the slope of the best average linear fit, and where $x_i$ ranges over the set $X_i$ of possible fractions of degree $i$.


We now proceed to consider each of the terms in $\hat{b}(\hat{\tau}_h) - b^*(\hat{\tau}_h)$ above. Each of the parentheses satisfies Donsker conditions. Indeed, under Conditions \ref{cond:positivity}, \ref{cond:bounded_var}, and \ref{cond:bounded_first}, they are sample-mean terms with uniformly bounded coefficients, and $\hat{\gamma}_n$ converges to $\gamma_n^*$ at a rate of $\mathcal{O}(n^{-1/2})$ (\citep{andrews1994empirical},see also Section 3.4.3.2. in \citep{van1996weak}). Thus, the terms in $\hat{M}_n^b(h) - M_n^b(h)$ have bounded entropy integral. Under our bounded-variation potential-outcome model (Condition \ref{cond:bounded_var}), we have that $\log N_{[]}(\epsilon, \Theta_n^b, L_2(\mathbb{P}) \leq K \left( \frac{1}{\epsilon} \right)$, which gives us $\Tilde{J}(\delta, \Theta_n^b, L_2(\mathbb{P}))  \leq \delta^{1/2}$ \citep{andrews1994empirical,van1996weak} (see also Example 19.11 in \citep{van2000asymptotic}).

Therefore, putting these together, we have the following maximal inequality for the bias terms
    \begin{align}  
    \label{maximal_ineq_proofb}
    &\mathbb{E}\left[\sup_{\delta/2 < |h^* - \hat{h}_n|<\delta} \sqrt{n} \left(\hat{M}_n^b(h_n) - M_n^b(h_n) - (\hat{M}_n^b(h^*) - M_n^b(h^*))\right)\right] \\
    &\leq  \Tilde{J}(\delta, \Theta_n^b, L_2(\mathbb{P})) \left(1+\frac{\Tilde{J}(\delta, \Theta_n^b, L_2(\mathbb{P})) }{ \delta \sqrt{n}}\right),
    \end{align} with $\Tilde{J}(\delta, \Theta_n^b, L_2(\mathbb{P}))  \leq \delta^{1/2}$.

\end{proof}

\end{document}